\documentclass[amsthm,twocolumn]{autart}

\usepackage{amsmath,amssymb}
\usepackage{graphicx}
\graphicspath{{}}
\usepackage{booktabs}
\usepackage{natbib}
\usepackage{tikz}
\usetikzlibrary{arrows.meta, decorations.pathreplacing}

\begin{document}

\begin{frontmatter}

\title{Hybrid Cones: Set Representation with Sharp Unions of Minimal Overhead for Nonlinear and Switched Linear Systems}

\author[nova,cts,isr]{Daniel Silvestre}\ead{dsilvestre@fct.unl.pt}

\address[nova]{School of Science and Technology, NOVA University of Lisbon, 2829-516 Caparica, Portugal.}
\address[cts]{Center for Technology and Systems (CTS), part of LASI, 2829-516 Caparica, Portugal.}
\address[isr]{Institute for Systems and Robotics, Av.Rovisco Pais, 1 1049-001 Lisboa, Portugal.}

\thanks[footnoteinfo]{This work was partially supported by FCT – Fundação para a Ciência e a Tecnologia, I.P., under project references UID/PRR/00066/2025 with DOI identifier \textrm{https://doi.org/10.54499/UID/PRR/00066/2025} and UID/00066/2025 with DOI identifier \textrm{https://doi.org/10.54499/UID/00066/2025}}


\begin{abstract}
	Reachability analysis and set-based estimation for nonlinear and hybrid systems require set representations capable of capturing nonconvex geometry. Hybrid Zonotopes (HZs) have recently gained attention because they augment constrained-zonotope generator descriptions with binary variables, allowing a single representation to encode a union of up to $2^n$ polytopic components. Recent work has further shown that several HZ operations can be made \textit{sharp}, meaning that relaxing the binary variables yields exactly the convex hull of the represented set. However, no known HZ construction provides a sharp union of two convex sets that adds only one binary generator and no further continuous generators or equality constraints to those of the operands. This limitation affects scalability in repeated union operations, ReLU neural-network reachability, and nonlinear propagation methods that cover graph sets by unions of simpler sets.
	
	This paper identifies a connection between the exact closed-form convex-hull operation for Constrained Convex Generators (CCGs) and the binary-generator structure underlying HZs. This connection leads to hybrid CCGs and to Hybrid Cones, a practical set representation capable of encoding up to $2^n$ norm-cone components. Hybrid Cones inherit the logical structure of HZs while allowing richer convex components, including polytopic and norm-based sets. They support sharp unions requiring only one added binary variable and can represent unions and intersections of polytopes and ellipsoids. For a union of $L$ sets built recursively with pairwise unions, the data structure of the complement of the equivalent HZ grows cubically in $L$, whereas the complement of the corresponding Hybrid Cone grows linearly. As a key application, we derive an exact Hybrid Cone representation of the ReLU graph using only two continuous generators and one binary generator. Numerical experiments compare Hybrid Cones with HZs in a localization problem where the nonlinear energy measurement function is approximated by a neural network and the case of a verification and validation procedure of the closed-loop of a double integrator corrupted by a nonlinear disturbance with a saturated Linear Quadratic Regulator (LQR) controller, showing substantial reductions in representation size and reduced conservatism in comparison with HZs.
\end{abstract}

\end{frontmatter}

\section{Introduction}
\label{sec:intro}

Set-based methods provide a natural framework for the analysis and control of dynamical systems in the presence of bounded disturbances, measurement noise, model uncertainty and actuator limits \cite{taborda:deadbandMPC}. Rather than estimating a single trajectory, they compute a guaranteed set containing all states or trajectories consistent with the available information. This approach underlies classical set-membership state estimation \cite{schweppe:68,chisci1996recursive}, reachability analysis \cite{althoff2014online}, and a wide range of robust and safety-related control applications. The effectiveness of these methods, however, depends critically on the set representation used to store and propagate the resulting sets. A useful representation must balance geometric expressiveness with the computational complexity of the operations required to manipulate and query it.

A substantial body of research has consequently examined set representations for this purpose. Early set-membership approaches considered intervals, parallelotopes, ellipsoids, and polytopes \cite{schweppe:68,chisci1996recursive,chernousko:05}. Zonotopes subsequently became particularly attractive because linear maps and Minkowski sums can be computed efficiently and exactly, while their generator structure can mitigate the wrapping effect that arises during repeated propagation \cite{combastel:03,camacho:05,althoff:comparison}. These properties have led to their widespread use in guaranteed state estimation and reachability analysis. Comparisons of set representations have shown that the underlying data structure can have a decisive impact on both accuracy and computational scalability \cite{althoff:comparison,silvestre:comparisonSetMethods}.

Constrained zonotopes extended this framework by introducing additional linear constraints on the generator variables. This enlarged the class of convex sets that can be represented while preserving a compact generator-based description and enabling exact formulations of operations such as intersection, projection, and Minkowski addition \cite{scott:cz}. Subsequent developments generalized this construction to richer classes of convex generators, including ellipsotopes and the more general Constrained Convex Generators (CCGs) \cite{ellipsotopes,silvestre:ccgs,silvestre:exactCCGhull}. In particular, CCGs provide a flexible framework for representing convex sets through generators subject to general convex constraints, thereby allowing different convex geometries to be handled within a common data structure. This framework also enables exact formulations of operations such as convex hulls, which are important when multiple convex branches must be combined \cite{silvestre:exactCCGhull}.

The increasing interest in expressive set representations is motivated by the fact that convexity itself can become a fundamental source of conservatism. In reachability analysis of hybrid systems, different modes may generate distinct reachable branches, while nonlinear mappings can transform a convex set into a nonconvex set with multiple disconnected or piecewise-defined regions. Replacing such a set by a convex enclosure may therefore introduce conservatism at every propagation step, with the resulting error accumulating over time. This motivates representations that preserve nonconvexity rather than repeatedly convexifying the reachable set.

Several approaches have therefore been developed to represent nonconvex sets directly. Hybrid Zonotopes (HZs), for example, augment the generator description of constrained zonotopes with binary variables \cite{bird:hz}. These variables encode alternative branches, allowing a single HZ to represent a union of constrained zonotopes while retaining the computational structure of a generator-based representation. Subsequent work has developed HZ operations for unions and complements \cite{bird:complement}, established sharpness properties for several operations \cite{glunt:sharp}, and applied HZs to nonlinear reachability, successor-set computation, functional decomposition, and set-valued state estimation \cite{siefert:nonlinearHZ,siefert:functionalHZ,siefert:estimationHZ}.

Other recent approaches have sought to preserve nonlinear dependencies directly. Sparse polynomial zonotopes and constrained polynomial zonotopes, for example, provide expressive nonconvex set representations that can capture nonlinear dependencies and support reachability computations for nonlinear systems \cite{kochdumper:spz,kochdumper:cpz}. These representations illustrate an important trade-off in nonconvex set-based analysis: increasing the expressive power of the representation can substantially reduce geometric conservatism, but the computational complexity of manipulating and querying the resulting set may increase accordingly. In particular, when set membership, intersection, containment, or emptiness cannot be reduced to the specialized linear or convex optimization problems associated with simpler generator representations, the corresponding operations may require solving general nonlinear optimization problems with integer or logical variables, which can become computationally expensive when solved repeatedly during reachability propagation, state estimation, or verification. Representing nonconvexity alone is therefore not sufficient. The computational structure of the representation must also support the queries required for state estimation and reachability analysis.

Hybrid Zonotopes and CCGs can be viewed as two complementary extensions of the constrained-zonotope paradigm. HZs extend the representation toward nonconvexity by introducing binary variables, but their convex components remain polytopic. CCGs extend the representation toward richer convex geometry, but remain convex. More general nonconvex representations, such as polynomial-zonotope-based approaches, can capture richer nonlinear structure, but may introduce substantially more difficult optimization problems for fundamental set queries \cite{kochdumper:spz,kochdumper:cpz}. Consequently, there remains a need for a representation that simultaneously provides general convex components, an explicit representation of nonconvexity, and computationally tractable operations and queries.

The need for such a representation is particularly evident in neural-network verification. ReLU networks are piecewise-affine, and exact verification methods commonly introduce binary variables to encode the activation state of individual neurons \cite{fischetti:relu}. Hybrid zonotopes have also been shown to exactly represent ReLU networks, using two constraints, four continuous generators and one binary generator per neuron \cite{ortiz:reluHZ}, and have been used more broadly to propagate nonconvex reachable sets through neural feedback systems \cite{zhang:xu:hzNN}. These results illustrate the central role played by the representation of the union induced by piecewise-affine structure of the ReLU, where a set representation that preserves the nonconvexity of the ReLU graph with a more economical generator count reduces the overall computation.

This paper introduces \emph{hybrid Constrained Convex Generators} (hybrid CCGs) as a unified extension of the CCG framework to nonconvex sets, obtained by augmenting the generator description of a CCG with the same type of binary variables used to construct HZs. Specializing this construction to norm-cone generators yields \emph{Hybrid Cones}, the practical specialization that we will develop in the paper. The construction extends the generator-based paradigm along both the geometric and logical dimensions with Hybrid Cones representing unions of convex sets with different geometric structures, including polytopic and norm-based components. Hybrid Zonotopes arise as a special case, while Hybrid Cones allow substantially richer convex sets to be used as the individual components.

A central property of the proposed representation is that unions can be represented \emph{sharply} with only one additional binary generator. Here, sharpness means that relaxing the binary variables in the resulting representation produces the convex hull of the represented nonconvex set. The construction arises from an identified connection between CCGs and their exact convex-hull operation \cite{silvestre:exactCCGhull}, which provides the foundation for extending the convex-generator framework to a hybrid setting while retaining a compact representation of unions. This is consistent with the broader literature on optimal formulations of polyhedral unions, where the number of auxiliary variables required is likewise the central measure of efficiency \cite{optimalCvxHull}. Because Hybrid Cones inherit this structure, fundamental operations can be formulated using the efficient algebraic and convex-analytic tools of the underlying generator representation, rather than a general nonlinear mixed-integer program.

The remainder of the paper develops the Hybrid Cone representation and its fundamental operations. We first define hybrid CCGs and specialize them to Hybrid Cones. We then establish the sharp union construction and its convex-hull relaxation property, and derive an exact compact representation of the ReLU graph. We also present the complement of HZs when written as Hybrid Cones and discuss the improvement in complexity growth for the novel set description. Numerical experiments compare Hybrid Cones with Hybrid Zonotopes in a set-based observer for a localization problem and in the verification of a closed-loop system with a saturated LQR controller. The results show that Hybrid Cones preserve nonconvexity and richer convex component geometry while reducing the representation size associated with repeated unions in the considered applications.

\section{Problem Formulation}
\label{sec:problem}

In this paper, we consider set-valued estimation and reachability analysis for nonlinear systems and discrete-time hybrid systems. Although the two classes of systems differ in the origin of their nonconvexity, their set-based forward propagation has a common structure. For hybrid systems, nonconvexity arises from the collection of admissible scheduling sequences. For nonlinear systems, it arises when the graph of the nonlinear dynamics is decomposed into a collection of convex leaves. In both cases, the resulting successor set can therefore be expressed as a finite union of convex sets. We do not discuss it explicitly, but the tool can also be used for fault detection and isolation resorting to the concept of distinguishability \cite{silvestre:distinguish}.

This common structure motivates the development of a set representation that can compactly encode unions of convex sets while retaining a rich description of each individual component. We first formulate the propagation problem for discrete-time hybrid systems and then show how nonlinear propagation can be expressed in the same form through a graph-set decomposition.

\subsection{Switched Linear Systems}
\label{sec:problem-hybrid}

Consider a discrete-time hybrid system with state $x_k \in \mathbb{R}^n$, disturbance $w_k \in \mathbb{R}^d$, and a finite collection of linear modes
\begin{equation}
	\mathcal{Q} = \{1,\ldots,N_q\},
	\label{eq:hybrid-mode-set}
\end{equation}
with mode-dependent matrices
\begin{equation}
	\left\{(A_q,B_q)\right\}_{q\in\mathcal{Q}}.
	\label{eq:hybrid-mode-matrices}
\end{equation}
For a switching signal $q_k\in\mathcal{Q}$, the system dynamics are
\begin{equation}
	x_{k+1}
	=
	A_{q_k}x_k+B_{q_k}w_k,
	\qquad
	w_k\in\mathcal{W},
	\label{eq:hybrid-dynamics}
\end{equation}
where $\mathcal{W}\subset\mathbb{R}^d$ is a known disturbance set.

Given a set of possible states $\mathcal{X}_k\subset\mathbb{R}^n$, the one-step successor set is
\begin{equation}
	\mathcal{X}_{k+1}
	=
	\bigcup_{q\in\mathcal{Q}}
	\left\{
	A_qx+B_qw:
	x\in\mathcal{X}_k,\;
	w\in\mathcal{W}
	\right\},
	\label{eq:hybrid-union}
\end{equation}
or equivalently,
\begin{equation}
	\mathcal{X}_{k+1}
	=
	\bigcup_{q\in\mathcal{Q}}
	\left(
	A_q\mathcal{X}_k+B_q\mathcal{W}
	\right).
	\label{eq:hybrid-union-affine}
\end{equation}
Each mode therefore produces a branch
\begin{equation}
	\mathcal{X}_{k+1}^{(q)}
	:=
	A_q\mathcal{X}_k+B_q\mathcal{W}.
	\label{eq:hybrid-one-step-branch}
\end{equation}
If $\mathcal{X}_k$ and $\mathcal{W}$ are convex, then every $\mathcal{X}_{k+1}^{(q)}$ is convex because it is the Minkowski sum of affine images of convex sets. The complete successor set, however, is generally nonconvex because it is a union of these mode-dependent branches.

Over multiple propagation steps, the relevant discrete variable is not an individual mode but a scheduling sequence. Define
\begin{equation}
	\mathbf{q}
	:=
	(q_0,q_1,\ldots,q_{k-1})
	\in\mathcal{Q}^k.
	\label{eq:scheduling-vector}
\end{equation}
For a fixed scheduling sequence $\mathbf{q}$, the state at time $k$ can be written as
\begin{equation}
	x_k
	=
	\Phi_{\mathbf{q}}x_0
	+
	\sum_{i=0}^{k-1}
	\Phi_{\mathbf{q}}^{(i+1)}B_{q_i}w_i,
	\label{eq:hybrid-k-step}
\end{equation}
where
\begin{equation}
	\Phi_{\mathbf{q}}
	:=
	A_{q_{k-1}}\cdots A_{q_1}A_{q_0}
	\label{eq:hybrid-transition-product}
\end{equation}
and
\begin{equation}
	\Phi_{\mathbf{q}}^{(i+1)}
	:=
	A_{q_{k-1}}\cdots A_{q_{i+1}},
	\label{eq:hybrid-transition-product-input}
\end{equation}
with the convention
\begin{equation}
	\Phi_{\mathbf{q}}^{(k)}=I.
	\label{eq:hybrid-transition-product-convention}
\end{equation}

Consequently, the $k$-step reachable set can be decomposed according to the possible scheduling sequences as
\begin{equation}
	\mathcal{X}_k
	=
	\bigcup_{\mathbf{q}\in\mathcal{Q}^k}
	\mathcal{X}_k^{(\mathbf{q})},
	\label{eq:hybrid-k-step-union}
\end{equation}
where
\begin{equation}
	\mathcal{X}_k^{(\mathbf{q})}
	=
	\Phi_{\mathbf{q}}\mathcal{X}_0
	+
	\sum_{i=0}^{k-1}
	\Phi_{\mathbf{q}}^{(i+1)}B_{q_i}\mathcal{W}.
	\label{eq:hybrid-schedule-branch}
\end{equation}
For convex $\mathcal{X}_0$ and $\mathcal{W}$, each fixed-schedule reachable set $\mathcal{X}_k^{(\mathbf{q})}$ is convex. The complete reachable set is therefore a finite union of convex sets indexed by the scheduling vectors in $\mathcal{Q}^k$.

This formulation exposes the central computational challenge for hybrid systems. Direct enumeration associates one convex set with every scheduling sequence. In the absence of additional structural restrictions, the number of possible sequences is
\begin{equation}
	\left|\mathcal{Q}^k\right|
	=
	N_q^k.
	\label{eq:hybrid-schedule-count}
\end{equation}
Consequently, repeated forward propagation can lead to exponential growth in the number of convex branches. Conversely, replacing the union in \eqref{eq:hybrid-k-step-union} by a single convex enclosure eliminates the explicit branch growth but may introduce substantial geometric conservatism.

The same structure appears in set-valued state estimation. Given a set-valued estimate $\mathcal{X}_k$, the prediction step generates the union in \eqref{eq:hybrid-union}. A subsequent measurement update may remove portions of individual branches or eliminate entire scheduling sequences. Hence, prediction and measurement updates can repeatedly create and manipulate nonconvex unions. A suitable representation should therefore encode the scheduling-induced branches directly, without requiring explicit enumeration or repeated convexification.

\subsection{Nonlinear Systems}
\label{sec:problem-nonlinear}

Consider now a discrete-time nonlinear system
\begin{equation}
	x_{k+1}
	=
	f(x_k,w_k),
	\qquad
	x_k\in\mathbb{R}^n,
	\qquad
	w_k\in\mathcal{W},
	\label{eq:nonlinear-dynamics}
\end{equation}
where
\begin{equation}
	f:\mathbb{R}^{n+d}\rightarrow\mathbb{R}^n
	\label{eq:nonlinear-function-domain}
\end{equation}
is a nonlinear map. Given a state set $\mathcal{X}_k$, the exact one-step reachable set is
\begin{equation}
	\mathcal{X}_{k+1}
	=
	f(\mathcal{X}_k,\mathcal{W})
	:=
	\left\{
	f(x,w):
	x\in\mathcal{X}_k,\;
	w\in\mathcal{W}
	\right\}.
	\label{eq:nonlinear-image}
\end{equation}
Even when $\mathcal{X}_k$ and $\mathcal{W}$ are convex, the nonlinear image in \eqref{eq:nonlinear-image} is generally nonconvex and cannot, in general, be represented as a single affine image of the original sets.

To expose the structure relevant to set-based propagation, consider the graph of the nonlinear dynamics over the domain
\begin{equation}
	\mathcal{D}
	:=
	\mathcal{X}_k\times\mathcal{W}.
	\label{eq:nonlinear-domain}
\end{equation}
The graph set is defined as
\begin{equation}
	\mathcal{G}_f(\mathcal{D})
	:=
	\left\{
	(x,w,y)\in\mathbb{R}^{2n+d}:
	(x,w)\in\mathcal{D},\;
	y=f(x,w)
	\right\}.
	\label{eq:graph-set}
\end{equation}
We can also construct a graph set $\Phi$ for an arbitrary domain $\mathcal{P}$ such that all possible subdomains of interest $\mathcal{D}$ satisfy $\mathcal{D} \subseteq \mathcal{P}$ and retrieve $\mathcal{G}_f(\mathcal{D})$ as $\mathcal{G}_f(\mathcal{D}) = \Phi \cap_{\Pi_p} \mathcal{D}$ where $p$ denotes the coordinates for $(x,w)$. This procedure is illustrated in Figure \ref{fig:leaf}.

The reachable set is then the projection of the graph set onto the output coordinates:
\begin{equation}
	\mathcal{X}_{k+1}
	=
	\operatorname{proj}_{y}
	\left(
	\mathcal{G}_f(\mathcal{D})
	\right),
	\label{eq:nonlinear-projection}
\end{equation}
where $\operatorname{proj}_{y}\left((x,w,y)\right) = \Pi_y(x,w,y) = y$ and similarly for the other variables. 

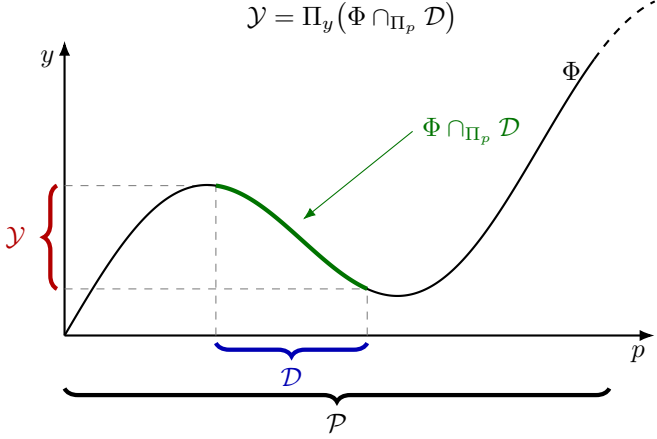
\begin{figure}
	\centering
	\begin{tikzpicture}[
		>=Latex,
		axis/.style={-{Latex[length=2mm]}, thick},
		curve/.style={thick},
		leaf/.style={ultra thick, green!45!black},
		guide/.style={dashed, gray, thin},
		domainbar/.style={ultra thick, blue!70!black},
		totaldomainbar/.style={ultra thick, black},
		imagebar/.style={ultra thick, red!70!black},
		]
		
		\pgfmathsetmacro{\pA}{2}
		\pgfmathsetmacro{\pB}{4}
		\pgfmathsetmacro{\pMid}{3}
		\pgfmathsetmacro{\fA}{0.4*\pA + 1.3*sin(\pA r)}   
		\pgfmathsetmacro{\fB}{0.4*\pB + 1.3*sin(\pB r)}   
		\pgfmathsetmacro{\fMid}{0.4*\pMid + 1.3*sin(\pMid r)} 
		
		\node at (3.8,4.2) {$\mathcal{Y} = \Pi_y\bigl(\Phi \cap_{\Pi_p} \mathcal{D}\bigr)$};
		
		\draw[axis] (0,0) -- (7.8,0) node[below left] {$p$};
		\draw[axis] (0,0) -- (0,3.9) node[below left] {$y$};
		
		\draw[guide] (0,\fA) -- (\pA,\fA);
		\draw[guide] (0,\fB) -- (\pB,\fB);
		\draw[guide] (\pA,0) -- (\pA,\fA);
		\draw[guide] (\pB,0) -- (\pB,\fB);
		
		\draw[curve, domain=0:7, samples=120, smooth]
		plot (\x, {0.4*\x + 1.3*sin(\x r)});
		\draw[curve, dashed, domain=7:7.8, samples=20, smooth]
		plot (\x, {0.4*\x + 1.3*sin(\x r)});
		\node[right] at (6.45,3.5) {$\Phi$};
		
		\draw[leaf, domain=\pA:\pB, samples=60, smooth]
		plot (\x, {0.4*\x + 1.3*sin(\x r)});
		
		\draw[decorate, decoration={brace, amplitude=5pt, mirror}, domainbar]
		(\pA,-0.1) -- (\pB,-0.1)
		node[midway, yshift=-0.45cm, blue!70!black] {$\mathcal{D}$};
		
		\draw[decorate, decoration={brace, amplitude=5pt, mirror}, totaldomainbar]
		(0,-0.7) -- (7.2,-0.7)
		node[midway, yshift=-0.45cm, black] {$\mathcal{P}$};
		
		\draw[decorate, decoration={brace, amplitude=5pt}, imagebar]
		(-0.1,\fB) -- (-0.1,\fA)
		node[midway, xshift=-0.55cm, red!70!black] {$\mathcal{Y}$};
		
		\draw[->, green!45!black] (4.6,2.7) -- (\pMid+0.15,\fMid+0.15);
		\node[right, green!45!black] at (4.6,2.7) {$\Phi \cap_{\Pi_p} \mathcal{D}$};
		
	\end{tikzpicture}
	\caption{Image computation through the graph set of a nonlinear function $\phi(\cdot)$. Restricting the graph set $\Phi$ constructed for the entire domain $\mathcal{P}$ to a domain interval $\mathcal{D}$ (blue) selects the leaf $\Phi \cap_{\Pi_p} \mathcal{D}$ (green). Projecting that leaf onto the $y$-axis recovers the exact image $\mathcal{Y} = \phi(\mathcal{D})$ (red).}
	\label{fig:leaf}
\end{figure}

The graph of a nonlinear function can be approximated or represented through a collection of local graph sets, which we refer to as leaves. Let $\mathcal L:=\{\mathcal L_\ell\}_{\ell\in\mathcal I_L}$ denote a collection of convex leaves covering the graph over the domain $\mathcal P$ introduced above,
\begin{equation}
	\mathcal G_f(\mathcal P)\subseteq\bigcup_{\ell\in\mathcal I_{L}}\mathcal L_{\ell}.
	\label{eq:graph-leaf-decomposition}
\end{equation}
The leaves may be generated by partitioning $\mathcal P$ and constructing a local representation of the graph over each partition. The precise procedure can be changed, but it is sufficient that the graph be covered by a finite collection of convex sets. Since $\mathcal G_f(\mathcal D)=\mathcal G_f(\mathcal P)\cap_{\Pi_p}\mathcal D$, the leaves must be restricted to $\mathcal D$ before projecting, which gives
\begin{align}
	\mathcal{X}_{k+1}&\subseteq\operatorname{proj}_{y}\Bigl(\bigcup_{\ell\in\mathcal I_{L}}\mathcal{L}_{\ell}\cap_{\Pi_p}\mathcal D\Bigr)
	=\bigcup_{\ell\in\mathcal I_{L}}\mathcal X_{k+1}^{(\ell)},\\
	\qquad
	\mathcal X_{k+1}^{(\ell)}&:=\operatorname{proj}_{y}\bigl(\mathcal L_\ell\cap_{\Pi_p}\mathcal D\bigr).
	\label{eq:nonlinear-convex-leaf-union}
\end{align}
If $\mathcal D$ is convex, as in the first propagation step from a convex initial set, every branch $\mathcal X_{k+1}^{(\ell)}$ is convex, because intersections and linear projections preserve convexity. In general $\mathcal X_k$ is itself a union, so $\mathcal D$ is nonconvex and the branches need not be. This is not an obstacle for Hybrid Cones, since the union of the leaves is stored once as a single Hybrid Cone and intersected once with $\mathcal D$, so each step adds the generators of the leaf union to those of $\mathcal X_k$.

The role of the leaves in nonlinear propagation is therefore analogous to the role of scheduling sequences in switched propagation. Therefore, the common problem considered in this paper is to develop a set representation capable of simultaneously providing a compact representation of finite unions of convex sets and a sufficiently expressive description of each convex component.

\section{Generator-Based Set Representations}
\label{sec:setRepresentations}

This section fixes the notation used throughout the paper. We first describe CCGs as affine images of convex generator variables. We then recall the hybrid zonotope construction and use the same binary-generator mechanism to define hybrid CCGs.

\subsection{Constrained convex generators}

\label{subsec:CCGs}
\begin{defn}[CCG \cite{silvestre:ccgs}]
	\label{def:ccg}
	A set $\mathcal{Z} \subset \mathbb{R}^n$ is a CCG if there exists a tuple 
	$(G, c, A, b) \in 
	\mathbb{R}^{n\times n_g}\times\mathbb{R}^n\times\mathbb{R}^{n_c\times n_g}\times\mathbb{R}^{n_c}$ 
	and a set $\mathfrak{G} = \mathcal{G}_1\times\mathcal{G}_2\times\dots\times\mathcal{G}_{N_g} 
	\subset \mathbb{R}^{n_g}$ such that
	\begin{equation}
		\mathcal{Z} = \{G\xi + c: 
		A\xi = b, \xi \in \mathfrak{G}\},
		\label{Eq:DefinitionCCG}
	\end{equation}
	where, for each $j \in \{1, \dots, N_g\}$, the generator set $\mathcal{G}_j$ is defined as the 
	zero-sublevel set of a convex function $g_j: \mathbb{R}^{n_j} \to \mathbb{R}^{m_j}$:
	\begin{equation}
		\mathcal{G}_j = \left\{\xi_j \in \mathbb{R}^{n_j}: g_j(\xi_j) \leq 0\right\},
		\label{Eq:CCGGeneratorSet}
	\end{equation}
	where the inequalities are taken component-wise.
\end{defn}

For a CCG $\mathcal{Z}$ defined as in \eqref{Eq:DefinitionCCG}, we introduce the shorthand 
notation $\mathcal{Z} = (G, c, A, b, \mathfrak{G}) \subset 
\mathbb{R}^n$. Given this formulation, 
\begin{itemize}
	\item an interval corresponds to $(G,c,[\,],[\,],\|\xi\|_{\infty} \leq 1)$, for a diagonal matrix $G$;
	\item a zonotope is given by $(G,c,[\,],[\,],\|\xi\|_{\infty} \leq 1)$;
	\item an ellipsoid is defined by $(G,c,[\,],[\,],\|\xi\|_{2} \leq 1)$, for a square matrix $G$;
	\item a CZ or polytope is $(G,c,A,b,\|\xi\|_{\infty} \leq 1)$;
	\item a finitely generated cone in $\mathbb{R}^{n}$ is $(G,0,[\,],[\,],\xi \geq 0)$;
	\item ellipsotopes are given by $(G,c,A,b,\|\xi\|_{p_{1}} \leq 1, \cdots$ \\ $, \|\xi\|_{p_{m}} \leq 1)$, for some $p_{i} \geq 1$, $1 \leq i \leq m$;
	\item AH-polytopes are given by $(G,c,[\,],[\,],H \xi \leq h)$.
\end{itemize}

We also recall some mathematical identities for the operations involving CCGs.

\setlength{\arraycolsep}{1pt}
\begin{prop}[Set Operations with CCGs \cite{silvestre:ccgs}]
	\label{prop:ccgOps}
	Consider a matrix $R \in \mathbb{R}^{y\times n}$, a vector 
	$t \in \mathbb{R}^y$, and three CCGs:
	\begin{itemize}
		\item $\mathcal{Z} = 
		(G_z, c_z, A_z, b_z, \mathfrak{G}_z)
		\subset \mathbb{R}^n$;
		\item $\mathcal{W} = 
		(G_w, c_w, A_w, b_w, \mathfrak{G}_w)
		\subset \mathbb{R}^n$;
		\item $\mathcal{V} = 
		(G_v, c_v, A_v, b_v, \mathfrak{G}_v)
		\subset \mathbb{R}^y$.
	\end{itemize}
	Then, regarding the affine map, Minkowski sum, and generalized intersection operations, the 
	following identities hold:
	\begin{equation}
		\begin{aligned}
			R\mathcal{Z} + t &= (RG_z, 
			Rc_z + t, A_z, b_z, \mathfrak{G}_z),\\
			\mathcal{Z} \oplus \mathcal{W} &= 
			\scalebox{0.85}{$\left(
				\begin{bmatrix}
					G_z & G_w
				\end{bmatrix}\hspace{-1mm},
				c_z + c_w, 
				\begin{bmatrix}
					A_z & \mathbf{0}\\
					\mathbf{0} & A_w
				\end{bmatrix}\hspace{-1mm},
				\begin{bmatrix}
					b_z\\
					b_w
				\end{bmatrix}\hspace{-1mm},
				\mathfrak{G}_z\times\mathfrak{G}_w\right)$},\\
			\mathcal{Z} \cap_R \mathcal{V} &= 
			\scalebox{0.79}{$\left(
				\begin{bmatrix}
					G_z & \mathbf{0}
				\end{bmatrix}\hspace{-1mm},
				c_z, 
				\begin{bmatrix}
					A_z & \mathbf{0}\\
					\mathbf{0} & A_v\\
					RG_z & -G_v
				\end{bmatrix}\hspace{-1mm},
				\begin{bmatrix}
					b_z\\
					b_v\\
					c_v - Rc_z
				\end{bmatrix}\hspace{-1mm},
				\mathfrak{G}_z\times\mathfrak{G}_v\right)$},\\
				\mathcal{Z} \cap \mathcal{P} &= 
				\scalebox{0.79}{$\left(\begin{bmatrix}
					G_z & \mathbf{0}
				\end{bmatrix}, c_z,\begin{bmatrix}
					A_z & \mathbf{0}\\
					p^\top G_z & d_m/2
				\end{bmatrix},\begin{bmatrix}
					b_z\\
					\frac{q - p^\top c_z - \|p^\top G_z\|_1}{2}
				\end{bmatrix},\mathfrak{G}_z \times \mathcal{B}_\infty\right)$},
		\end{aligned}
		\label{Eq:SetOperationsCCGs}
	\end{equation}
	where $\mathcal{P}=\{x \mid p^\top x \leq q\}\subset\mathbb{R}^n$ is a halfspace. For the last identity we assume that every generator set in $\mathfrak G_z$ is contained in the unit $\ell_\infty$ ball, as for zonotopes, constrained zonotopes, ellipsoids and ellipsotopes, and set $d_m=q-p^\top c_z+\|p^\top G_z\|_1$. If $d_m<0$ then $\mathcal Z\cap\mathcal P=\emptyset$, and otherwise the identity holds.
\end{prop}
\setlength{\arraycolsep}{2pt}
The assumption makes $d_m$ an upper bound on $q-p^\top c_z-p^\top G_z\xi$ over $\mathfrak G_z$, which is the range that the slack generator must cover. When $d_m<0$ the halfspace is violated at every point. For a bounded generator set not contained in the ball, $\|p^\top G_z\|_1$ is replaced by $\max_{\xi\in\mathfrak G_z}(-p^\top G_z\xi)$ in both $d_m$ and the right-hand side.

\subsection{Hybrid zonotopes and Hybrid Constrained Convex Generators}
\label{subsec:HZsHCCGs}

Hybrid zonotopes extend constrained zonotopes by adding binary generator variables. We use the binary convention
\begin{equation}
	\beta\in\{-1,1\}^{n_b}.
\end{equation}

\begin{defn}[Hybrid zonotope]
	\label{def:HZ}
	A hybrid zonotope is the set
	\begin{equation}
		\begin{aligned}
			\mathcal{H}
			=
			\big\{
			G^c\xi+G^b\beta+c:
			A^c\xi+&A^b\beta=b,\\
			&\xi\in\mathcal{B}_\infty^{n_{gc}},\;
			\beta\in\{-1,1\}^{n_b}
			\big\},
		\end{aligned}
	\end{equation}
	where $G^c\in\mathbb{R}^{n\times n_{gc}}$, $G^b\in\mathbb{R}^{n\times n_b}$, $A^c\in\mathbb{R}^{n_{c}\times n_{gc}}$, $G^b\in\mathbb{R}^{n_{c}\times n_b}$, and $c\in\mathbb{R}^{n}$.
\end{defn}

For each fixed binary assignment, the set in Definition~\ref{def:HZ} is a constrained zonotope. Hence,
\begin{equation}
	\scalebox{0.9}{$
	\mathcal{H}
	=
	\bigcup_{\beta\in\{-1,1\}^{n_b}}
	\left\{
	G^c\xi+G^b\beta+c:
	A^c\xi=b-A^b\beta,\;
	\xi\in\mathcal{B}_\infty^{n_c}
	\right\},$}
\end{equation}
which means a hybrid zonotope represents a union of at most $2^{n_b}$ polytopes.

The relaxation of a hybrid zonotope is obtained by replacing the binary domain by its convex hull:
\begin{equation}
	\operatorname{relax}(\mathcal{H})
	=
	\scalebox{0.8}{$
	\left\{
	\begin{bmatrix}
		G^c & G^b
	\end{bmatrix}
	\begin{bmatrix}
		\xi \\ \beta
	\end{bmatrix}+c:
	\begin{bmatrix}
		A^c & A^b
	\end{bmatrix}
	\begin{bmatrix}
		\xi \\ \beta
	\end{bmatrix}=b,
	\left\|\begin{bmatrix}
		\xi \\ \beta
	\end{bmatrix}\right\|_\infty \leq 1\right\}$}.
\end{equation}
A hybrid-zonotope representation is called sharp when
\begin{equation}
	\operatorname{relax}(\mathcal{H})
	=
	\operatorname{conv}(\mathcal{H}).
\end{equation}

The basic operations mirror the CCG expressions as we can consider $\xi$ and $\beta$ as $\xi_1$ and $\xi_2$ and have $g_1(\cdot)$ implementing the $\ell_\infty$ unit ball whereas $g_2(\cdot)$ constrains the variable to be in $\{-1,1\}$. We could therefore directly generalize HZs by considering CCGs that could also contain binary variables. 

\begin{defn}[Hybrid CCG]
	\label{def:HCCG}
	A \emph{Hybrid constrained convex generator} (Hybrid CCG) in $\mathbb{R}^n$ is a set of the form
	\begin{equation}
		\mathcal{HC} = \left\{ G \xi + c \;:\; A \xi = b,\; \xi \in \mathcal{C}_1 \times \cdots \times \mathcal{C}_p \right\},
	\end{equation}
	where $G \in \mathbb{R}^{n \times n_g}$, $c \in \mathbb{R}^n$, $A \in \mathbb{R}^{n_c \times n_g}$, $b \in \mathbb{R}^{n_c}$, and each factor $\mathcal{C}_i$ is either
	\begin{enumerate}
		\item a convex set representable as $\{\xi_i : g_i(\xi_i) \le 0\}$ for some convex function $g_i$ with $g_i(0) \le 0$ (in particular, $0 \in \mathcal{C}_i$), or
		\item a binary factor $\{-1,+1\}$.
	\end{enumerate}
\end{defn}

For each fixed binary vector, Definition~\ref{def:HCCG} produces a CCG, meaning that a hybrid CCG represents a union of at most $2^{n_b}$ CCG components, with HZs being recovered when $\mathcal{C}$ defines the cartesian product between unit $\ell_\infty$ norm ball or the sets $\{-1,1\}$.

\section{Hybrid Cones}
\label{sec:hybridCones}

We now specialize hybrid CCGs to a structured class of convex and binary generator constraints. The key difference from the CCG notation in Definition~\ref{def:ccg} is that the admissible generator set is no longer described as a Cartesian product of independent primitive sets. Instead, Hybrid Cones use a single generator vector subject to a list of constraints. Each constraint may involve selected, and possibly repeated, entries of the generator vector. This shared-generator structure is what permits the sharp union construction developed below. To shorten the notation, we will assume an index set $\mathcal{I} = \{1,\cdots,N_g\}$ for the $N_g$ continuous generator sets. 

\begin{defn}[Hybrid Cone]
	\label{def:hybridCone}
	A set is a Hybrid Cone $\mathcal{HC} \subset\mathbb{R}^{n}$ if there exists a tuple $\left(G,c,A,b\right) \in \mathbb{R}^{n\times n_g} \times \mathbb{R}^{n} \times \mathbb{R}^{n_c\times n_g} \times \mathbb{R}^{n_c}$ and a set $\mathfrak{G} = \{\mathcal{C}_i\}_{i=1}^{N_g+n_b}$ where each element $\mathcal{C}_i$ is of the form of a cone constraint $\|S_i\xi\|_{p_i}+h_i^\top \xi \leq d_i, \forall i \in \mathcal{I}$ or a binary constraint $E_i \xi \in \{-1,1\}, \forall i \in \{N_g+1, \cdots, N_g + n_b\}$ such that
	\begin{equation}
		\begin{aligned}
			\mathcal{HC}
			=
			\big\{
			G\xi+c:
			A\xi=b,\;
			\{\|S_i\xi\|_{p_i}+h_i^\top &\xi \leq d_i\}_{i \in \mathcal{I}},\\
			&E\xi\in\{-1,1\}^{n_b}\big\},
		\end{aligned}
	\end{equation}
	where $p_i\geq 1$ for all $i\in\mathcal{I}$, $S_i$ are projection matrices extracting continuous generator entries of $\xi$ and $E$ is a vertical concatenation of $E_i$ that selects all the binary generators.
\end{defn}

For a HC $\mathcal{HC}$ defined as in Definition \ref{def:hybridCone}, we associate the shorthand notation $\mathcal{HC} = (G, c, A, b, \mathfrak{G}) \subset \mathbb{R}^n$. We split the generator vector as $\xi=(\xi^{\mathrm c},\xi^{\mathrm b})$, with $\xi^{\mathrm b}=E\xi$ the binary generators (after reordering so that they come last), and partition $G=[G^{\mathrm c}\ G^{\mathrm b}]$, $A=[A^{\mathrm c}\ A^{\mathrm b}]$ and $h_i=(h_i^{\mathrm c},h_i^{\mathrm b})$ accordingly. From Definition \ref{def:hybridCone}, a set is described by a total of $n_g$ generators with $n_b$ binary generators, $n_c$ linear constraints and $N_g$ cone constraints. For simplicity, we will refer to $n_{gc}$ as the number of continuous generators and we draw attention to the fact that the cone constraints can have a linear dependence on both continuous and binary generators.

Equivalently, define the continuous cone-feasible generator set
\begin{equation}
	\mathcal{K}
	:=
	\left\{
	\xi\in\mathbb{R}^{n_g}:
	A\xi=b,\;
	\|S_i\xi\|_{p_i}+h_i^\top \xi \leq d_i,\;
	i\in \mathcal{I}
	\right\}
\end{equation}
and the binary-feasible generator set
\begin{equation}
	\mathcal{B}
	:=
	\left\{
	\xi\in\mathbb{R}^{n_g}:
	E\xi\in\{-1,1\}^{n_b}
	\right\}.
\end{equation}
Then the Hybrid Cone can be written compactly as
\begin{equation}
	\mathcal{HC}
	=
	\left\{
	G\xi+c:
	\xi\in\mathcal{K}\cap\mathcal{B}
	\right\}.
\end{equation}

The continuous relaxation of a Hybrid Cone is obtained by replacing each binary constraint with its convex hull as 
\begin{equation}
	\begin{aligned}
		\operatorname{relax}(\mathcal{HC})
		=
		\bigg\{
		G\xi+c:
		A\xi=b,\;
		\{\|S_i\xi\|_{p_i}+&h_i^\top \xi \leq d_i\}_{i \in \mathcal{I}}, \\
		&\|E\xi\|_\infty \leq 1\bigg\}
	\end{aligned}
\end{equation}

Hybrid Zonotopes are recovered as a special case by using box constraints on the continuous generators and binary constraints on the binary generators. In particular, with $\xi = \begin{bmatrix}
	\xi_c^\top & \xi_b^\top
\end{bmatrix}^\top$ and the matrices $I_c$, $I_b$ such that $\operatorname{blkdiag}\left(I_c,I_b\right) = I_{n_g}$, the HZ $\mathcal{H}$ defined by the tuple $\left(G_c,G_b,c,A_c,A_b,b\right)$ becomes the Hybrid Cone $\mathcal{HC}$ defined by the tuple
\begin{equation}
	\left(\begin{bmatrix}
		G^c & G^b
	\end{bmatrix},
	c,
	\begin{bmatrix}
		A^c & A^b
	\end{bmatrix},
	b, \big\{\|I_c\xi\|_{\infty} \leq 1,I_b\xi \in \{-1,1\}\big\}\right).
\end{equation}

\section{Sharp Union of Hybrid Cones}
\label{sec:union}

In order to present the union identity for Hybrid Cones, we require that no continuous generator can grow without bound while all the cone constraints remain satisfied. This rules out nonzero recession directions in the inactive branch of a union.

\begin{assum}
	\label{ass:bounded}
	For each Hybrid Cone to which the union is applied, the only continuous generator vector $\xi^{\mathrm c}$ satisfying
	\begin{equation}
		A^{\mathrm c}\xi^{\mathrm c}=0,\qquad
		\|S_i\xi^{\mathrm c}\|_{p_i}+(h_i^{\mathrm c})^\top\xi^{\mathrm c}\le 0,\quad i\in\mathcal I,
	\end{equation}
	is $\xi^{\mathrm c}=0$.
\end{assum}

Let us present the union of HCs with no binary generators using the perspective scaling from \cite{balas:perspectiveScaling} that will be helpful in the proof of the general result.

\begin{prop}[Sharp union, binary-free operands]
	\label{prop:union-base}
	Let $\mathcal X,\mathcal Y\subset\mathbb R^n$ be Hybrid Cones with $n_{b,\mathcal X}=n_{b,\mathcal Y}=0$, defined by the tuples
	\begin{equation}
		\mathcal X=\bigl(G_{\mathcal X},c_{\mathcal X},A_{\mathcal X},b_{\mathcal X},\mathfrak G_{\mathcal X}\bigr),\quad
		\mathcal Y=\bigl(G_{\mathcal Y},c_{\mathcal Y},A_{\mathcal Y},b_{\mathcal Y},\mathfrak G_{\mathcal Y}\bigr),
	\end{equation}
	where $\mathfrak G_{\mathcal X}=\{\|S_i\xi_{\mathcal X}\|_{p_i}+h_i^\top\xi_{\mathcal X}\le d_i\}_{i\in\mathcal I_{\mathcal X}}$
	and $\mathfrak G_{\mathcal Y}=\{\|S_j\xi_{\mathcal Y}\|_{p_j}+h_j^\top\xi_{\mathcal Y}\le d_j\}_{j\in\mathcal I_{\mathcal Y}}$,
	$\mathcal I_{\mathcal X}=\{1,\dots,N_g^{\mathcal X}\}$, $\mathcal I_{\mathcal Y}=\{1,\dots,N_g^{\mathcal Y}\}$,
	contain the (purely continuous) cone constraints of $\mathcal X$ and $\mathcal Y$, since
	neither set has binary generators.
	
	Let $\sigma$ be a scalar binary generator, $\alpha:=\tfrac{1+\sigma}2$, $\beta:=1-\alpha=\tfrac{1-\sigma}2$,
	and define $\mathcal{HC}\subset\mathbb R^n$ with generator vector $\xi=(\xi_{\mathcal X},\xi_{\mathcal Y},\sigma)$.
	The Hybrid Cone $\mathcal{HC}=(G,c,A,b,\{\mathcal C_i\}_{i=1}^{N_g^{\mathcal X}+N_g^{\mathcal Y}+1})$
	resulting from $\mathcal X\cup\mathcal Y$ is given by
	\begin{equation}
		G=\begin{bmatrix}G_{\mathcal X}&G_{\mathcal Y}&\tfrac12(c_{\mathcal X}-c_{\mathcal Y})\end{bmatrix},\qquad
		c=\tfrac12(c_{\mathcal X}+c_{\mathcal Y}),
	\end{equation}
	\begin{equation}
		A=\begin{bmatrix}A_{\mathcal X}&0&-\tfrac12 b_{\mathcal X}\\[2pt] 0&A_{\mathcal Y}&\tfrac12 b_{\mathcal Y}\end{bmatrix},
		\qquad
		b=\begin{bmatrix}\tfrac12 b_{\mathcal X}\\[2pt] \tfrac12 b_{\mathcal Y}\end{bmatrix},
	\end{equation}
	obtained from $A_{\mathcal X}\xi_{\mathcal X}=\alpha b_{\mathcal X}$ and $A_{\mathcal Y}\xi_{\mathcal Y}=\beta b_{\mathcal Y}$, 
	and with cone-constraint family
	\begin{equation}
		\begin{aligned}
			\{\mathcal C_i\}_{i=1}^{N_g^{\mathcal X}+N_g^{\mathcal Y}+1} &= \Bigl\{\|S_i\xi_{\mathcal X}\|_{p_i}+h_i^\top\xi_{\mathcal X}-\tfrac12 d_i\,\sigma\le\tfrac12 d_i\Bigr\}_{i\in\mathcal I_{\mathcal X}}\\
			&\cup\,
			\Bigl\{\|S_j\xi_{\mathcal Y}\|_{p_j}+h_j^\top\xi_{\mathcal Y}+\tfrac12 d_j\,\sigma\le\tfrac12 d_j\Bigr\}_{j\in\mathcal I_{\mathcal Y}}\\
			&\cup\,
			\bigl\{\sigma\in\{-1,1\}\bigr\}.
		\end{aligned}
	\end{equation}
	Then, under Assumption~\ref{ass:bounded},
	\begin{align}
		&\text{(i) exact union: }\ \mathcal{HC}=\mathcal X\cup\mathcal Y,\\
		&\text{(ii) sharp relaxation: }\ \operatorname{relax}(\mathcal{HC})=\operatorname{conv}(\mathcal X\cup\mathcal Y).
	\end{align}
\end{prop}

\begin{pf}
	\emph{(i), $\mathcal{HC}\supseteq \mathcal{X} \cup \mathcal{Y}$:} Let us pick a point $x\in\mathcal X$, which using $\tilde\xi_{\mathcal X}$ to represent the generator variables from $\mathcal{X}$, means that $x=G_{\mathcal X}\tilde\xi_{\mathcal X}+c_{\mathcal X}$, $\tilde\xi_{\mathcal X}\in\mathcal K_{\mathcal X}$. Since $\sigma=1$ ($\alpha=1,\beta=0$), $\xi_{\mathcal X}=\tilde\xi_{\mathcal X}$, $\xi_{\mathcal Y}=0$ since $\xi_{\mathcal{X}}$ is a perspective scaled version of $\tilde \xi_{\mathcal{X}}$ while $\xi_{\mathcal{Y}}$ is scaled by $\beta$. Every $\mathcal X$-constraint reduces to its original form, so it holds since $\tilde\xi_{\mathcal X}\in\mathcal K_{\mathcal X}$. Every $\mathcal Y$-constraint reduces to $A_{\mathcal{Y}}\xi_{\mathcal Y}=0$, and the cone constraints to $\|S_j\xi_{\mathcal Y}\|_{p_j}+h_j^\top\xi_{\mathcal Y}\le0$, holding trivially at $\xi_{\mathcal{Y}} = 0$. The point map gives $G\xi+c=G_{\mathcal X}\tilde\xi_{\mathcal X}+\tfrac12(c_{\mathcal X}-c_{\mathcal Y})+\tfrac12(c_{\mathcal X}+c_{\mathcal Y})=G_{\mathcal X}\tilde\xi_{\mathcal X}+c_{\mathcal X}=x$. The case $x\in\mathcal Y$ is symmetric at $\sigma=-1$, which means that all points in $\mathcal{X}$ and $\mathcal{Y}$ belong to $\mathcal{HC}$.
	
	\emph{(i),  $\mathcal{X} \cup \mathcal{Y} \supseteq \mathcal{HC}$:} Let $\xi=(\xi_{\mathcal X},\xi_{\mathcal Y},\sigma)$ be feasible with $\sigma=1$. The $\mathcal X$-block constraints reduce to $\xi_{\mathcal X}\in\mathcal K_{\mathcal X}$ exactly. The $\mathcal Y$-block constraints reduce to $A_{\mathcal Y}\xi_{\mathcal Y}=0$ and $\|S_j\xi_{\mathcal Y}\|_{p_j}+h_j^\top\xi_{\mathcal Y}\le0$ for all $j\in\mathcal I_{\mathcal Y}$. Since $\mathcal Y$ has no binary generators, $\xi_{\mathcal Y}=\xi_{\mathcal Y}^{\mathrm c}$ and these are exactly the conditions of Assumption~\ref{ass:bounded}, so $\xi_{\mathcal Y}=0$. The point map gives $x=G_{\mathcal X}\xi_{\mathcal X}+c_{\mathcal X}\in\mathcal X$. The case $\sigma=-1$ is symmetric, meaning that any point in $\mathcal{HC}$ is either in $\mathcal{X}$ when the binary variable $\sigma = 1$ or in $\mathcal{Y}$ when $\sigma = -1$.
	
	\emph{(ii):} Relaxing $\sigma\in\{-1,1\}$ to $\sigma\in[-1,1]$ leaves $\alpha,\beta\ge0$ and $\alpha+\beta=1$. For fixed such $(\alpha,\beta)$, the constraint set for $\xi_{\mathcal X}$ is exactly $\alpha\mathcal K_{\mathcal X}$ and for $\xi_{\mathcal Y}$ exactly $\beta\mathcal K_{\mathcal Y}$, so the resulting point is $\alpha(G_{\mathcal X}\eta_{\mathcal X}+c_{\mathcal X})+\beta(G_{\mathcal Y}\eta_{\mathcal Y}+c_{\mathcal Y})$ (at $\alpha = 0$ or $\beta = 0$ this uses Assumption \ref{ass:bounded} as in the proof of (i)). For some $\eta_{\mathcal X}\in\mathcal K_{\mathcal X}$ and $\eta_{\mathcal Y}\in\mathcal K_{\mathcal Y}$, we can write $\xi_{\mathcal X}=\alpha\eta_{\mathcal X}$ and $\xi_{\mathcal Y}=\beta\eta_{\mathcal Y}$, meaning that a point in the set corresponds to an arbitrary convex combination $\alpha x_1+\beta x_2$ for $x_1\in\mathcal X,x_2\in\mathcal Y$. Ranging over $\sigma\in[-1,1]$ sweeps out exactly $\{\alpha x_1+\beta x_2:\alpha,\beta\ge0,\alpha+\beta=1,x_1\in\mathcal X,x_2\in\mathcal Y\}=\operatorname{conv}(\mathcal X\cup\mathcal Y)$. \qed
\end{pf}

Proposition \ref{prop:union-base} exploits a known fact about CCGs since the perspective scaling was used in \cite{silvestre:exactCCGhull} to provide a closed-form exact expression for the convex hull between two CCGs. What the above result is providing is the equivalent but when the sets are written in HC format and by identifying that if we take the $\sigma$ variable to be in the extreme values we get the original operand sets. This slider effect of $\sigma$ is illustrated in Figure \ref{fig:sharp-union-example} where the union of two squares is shown and evaluated at $\sigma \in \{-1, -0.5, 0, 0.5, 1\}$.

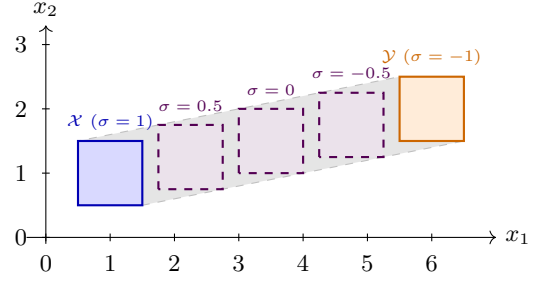
\begin{figure}
	\centering
	\begin{tikzpicture}[scale=0.85, every node/.style={font=\small}]
		
		\draw[gray!50, dashed, thin,fill=gray,fill opacity=0.2] (0.5,0.5) -- (1.5,0.5) -- (6.5,1.5) -- (6.5,2.5) -- (5.5,2.5) -- 
		(0.5,1.5) -- cycle;
		
		\draw[->] (-0.3,0) -- (7,0) node[right] {$x_1$};
		\draw[->] (0,-0.1) -- (0,3.3) node[above] {$x_2$};
		\foreach \x in {0,1,2,3,4,5,6}
		\draw (\x,0.06) -- (\x,-0.06) node[below=2pt] {\x};
		\foreach \y in {0,1,2,3}
		\draw (0.06,\y) -- (-0.06,\y) node[left=2pt] {\y};
		
		\draw[fill=blue!15, draw=blue!70!black, thick] (1/2,1/2) rectangle (1.5,1.5);
		\node[blue!70!black] at (1,1.8) {\tiny$\mathcal{X}\ (\sigma=1)$};
		
		\draw[fill=orange!15, draw=orange!80!black, thick] (5.5,1.5) rectangle (6.5,2.5);
		\node[orange!80!black] at (6,2.8) {\tiny$\mathcal{Y}\ (\sigma=-1)$};
		
		\draw[fill=violet!12, fill opacity=0.6, draw=violet!65!black, thick, dashed]
		(1.75,0.75) rectangle (2.75,1.75);
		\draw[fill=violet!12, fill opacity=0.6, draw=violet!65!black, thick, dashed]
		(3,1) rectangle (4,2);
		\draw[fill=violet!12, fill opacity=0.6, draw=violet!65!black, thick, dashed]
		(4.25,1.25) rectangle (5.25,2.25);
		
		\node[violet!65!black] at (2.25,2.05) {\tiny$\sigma=0.5$};
		\node[violet!65!black] at (3.5,2.3) {\tiny$\sigma=0$};
		\node[violet!65!black] at (4.75,2.55) {\tiny$\sigma=-0.5$};
	\end{tikzpicture}
	\caption{Sharp union $\mathcal{HC}$ of two Hybrid Cones $\mathcal{X}$ (blue, $\sigma=1$) and $\mathcal{Y}$ (orange, $\sigma=-1$). The three dashed violet boxes show $\operatorname{relax}(\mathcal{HC})$ evaluated at intermediate values $\sigma \in \{-0.5,0,0.5\}$ where $\sigma$ perspective-scales the two branches by complementary factors $\alpha=(1{+}\sigma)/2$ and $\beta=(1{-}\sigma)/2=1{-}\alpha$. Sweeping $\sigma$ continuously over $[-1,1]$ traces out exactly $\operatorname{conv}(\mathcal{X}\cup \mathcal{Y})$.}
	\label{fig:sharp-union-example}
\end{figure}

Let us now present the general result in the following theorem where we will denote $\xi^{\mathrm{c}}$ to identify the continuous generators among the entire generator vector.

\begin{thm}[Sharp union, general case]
	\label{thm:union-general}
	Let $\mathcal X,\mathcal Y\subset\mathbb R^n$ be arbitrary Hybrid Cones, defined by tuples
	$\mathcal X=(G_{\mathcal X},c_{\mathcal X},A_{\mathcal X},b_{\mathcal X},\mathfrak G_{\mathcal X})$,
	$\mathcal Y=(G_{\mathcal Y},c_{\mathcal Y},A_{\mathcal Y},b_{\mathcal Y},\mathfrak G_{\mathcal Y})$,
	with binary generators $\delta_{\mathcal X},\delta_{\mathcal Y}$. Let $\sigma$ be a scalar
	binary generator, $\alpha:=(1+\sigma)/2$, $\beta:=1-\alpha$, and define $\mathcal{HC}$ with
	generator vector $\xi=(\xi_{\mathcal X},\xi_{\mathcal Y},\sigma)$, where the binary
	sub-vectors of $\xi_{\mathcal X},\xi_{\mathcal Y}$ are now written $q_{\mathcal X}:=E_{\mathcal X}\xi_{\mathcal X}$,
	$q_{\mathcal Y}:=E_{\mathcal Y}\xi_{\mathcal Y}$, by
	\begin{equation}
		\begin{aligned}
			G=\begin{bmatrix}G_{\mathcal X}&G_{\mathcal Y}&\gamma\end{bmatrix}, \quad c=\tfrac12\bigl[(c_{\mathcal X}+g_{\mathcal X})+(c_{\mathcal Y}+g_{\mathcal Y})\bigr]\\
			\gamma:=\tfrac12\bigl[(c_{\mathcal X}-g_{\mathcal X})-(c_{\mathcal Y}-g_{\mathcal Y})\bigr],
		\end{aligned}
	\end{equation}
	\begin{equation}
		A=\begin{bmatrix}A_{\mathcal X}&0&-\tfrac12(b_{\mathcal X}+e_{\mathcal X})\\
			0&A_{\mathcal Y}&\tfrac12(b_{\mathcal Y}+e_{\mathcal Y})\end{bmatrix},
		\qquad
		b=\begin{bmatrix}\tfrac12(b_{\mathcal X}-e_{\mathcal X})\\ \tfrac12(b_{\mathcal Y}-e_{\mathcal Y})\end{bmatrix},
	\end{equation}
	where
	\begin{equation}
		\begin{aligned}
			g_{\mathcal{X}} = G_{\mathcal{X}}^{\mathrm{b}} \mathbf{1}, \quad 
			g_{\mathcal{Y}} = G_{\mathcal{Y}}^{\mathrm{b}} \mathbf{1}, \quad
			e_{\mathcal{X}} = A_{\mathcal{X}}^{\mathrm{b}} \mathbf{1}, \quad 
			e_{\mathcal{Y}} = A_{\mathcal{Y}}^{\mathrm{b}} \mathbf{1}
		\end{aligned}
	\end{equation}
	and 
	\begin{equation}
		\begin{aligned}
			\mathfrak G&=
			\Bigl\{\|S_i\xi^{\mathrm{c}}_{\mathcal X}\|_{p_i}+h_i^\top\xi_{\mathcal X}-\tfrac{\sigma}2(d_i{+}f_i)\le\tfrac12(d_i{-}f_i)\Bigr\}_{i\in\mathcal I_{\mathcal X}}\\
			&\cup
			\Bigl\{\|S_j\xi^{\mathrm{c}}_{\mathcal Y}\|_{p_j}+h_j^\top\xi_{\mathcal Y}+\tfrac{\sigma}2(d_j{+}f_j)\le\tfrac12(d_j{-}f_j)\Bigr\}_{j\in\mathcal I_{\mathcal Y}}\\
			&\cup\ \bigl\{q_{\mathcal X}\le\sigma\mathbf1,\,q_{\mathcal X}\in\{-1,1\}^{n_{b,\mathcal X}}\bigr\}\\
			&\cup\bigl\{q_{\mathcal Y}\le-\sigma\mathbf1,\,q_{\mathcal Y}\in\{-1,1\}^{n_{b,\mathcal Y}}\bigr\}\\
			&\cup\bigl\{\sigma\in\{-1,1\}\bigr\},
		\end{aligned}
	\end{equation}
	with 
	\begin{equation}
		f_i = \mathbf{1}^\top h_i^{\mathrm{b}} , i \in \mathcal{I}_{\mathcal{X}}, \quad f_j = \mathbf{1}^\top h_j^{\mathrm{b}}, j \in \mathcal{I}_{\mathcal{Y}}.
	\end{equation}
	Then, under Assumption~\ref{ass:bounded}, $\mathcal{HC}=\mathcal X\cup\mathcal Y$. If in
	addition $\operatorname{relax}(\mathcal X)=\operatorname{conv}(\mathcal X)$ and
	$\operatorname{relax}(\mathcal Y)=\operatorname{conv}(\mathcal Y)$, then
	$\operatorname{relax}(\mathcal{HC})=\operatorname{conv}(\mathcal X\cup\mathcal Y)$.
\end{thm}

\begin{pf}
	\emph{(i), $\mathcal{HC}\supseteq\mathcal X\cup\mathcal Y$:} Let $x\in\mathcal X$, $x=G_{\mathcal X}\tilde\xi_{\mathcal X}+c_{\mathcal X}$
	for some $\tilde\xi_{\mathcal X}=(\xi_{\mathcal X}^{\mathrm c},\delta_{\mathcal X})\in\mathcal K_{\mathcal X}\cap\mathcal B_{\mathcal X}$.
	Take $\sigma=1$ ($\alpha=1,\beta=0$), $\xi_{\mathcal X}=\tilde\xi_{\mathcal X}$ (so $q_{\mathcal X}=\delta_{\mathcal X}$),
	and $\xi_{\mathcal Y}^{\mathrm c}=0$, $q_{\mathcal Y}=-\mathbf1$. Substituting into the linear equation of 
	$\mathcal X$ it becomes $A_{\mathcal X}^{\mathrm c}\xi_{\mathcal X}^{\mathrm c}+A_{\mathcal X}^{\mathrm b}\delta_{\mathcal X}-\tfrac12(b_{\mathcal X}+e_{\mathcal X})=A_{\mathcal X}\tilde\xi_{\mathcal X}-\tfrac12(b_{\mathcal X}+e_{\mathcal X})=b_{\mathcal X}-\tfrac12(b_{\mathcal X}+e_{\mathcal X})=\tfrac12(b_{\mathcal X}-e_{\mathcal X})$,
	matching $b$ exactly since $A_{\mathcal X}\tilde\xi_{\mathcal X}=b_{\mathcal X}$. Equivalently, the constraints in $\mathcal Y$
	becomes $A_{\mathcal Y}^{\mathrm b}(-\mathbf1)+\tfrac12(b_{\mathcal Y}+e_{\mathcal Y})=-e_{\mathcal Y}+\tfrac12(b_{\mathcal Y}+e_{\mathcal Y})=\tfrac12(b_{\mathcal Y}-e_{\mathcal Y})$,
	matching $b$ trivially. The same substitution into each $\mathfrak G$-inequality reduces
	the $\mathcal X$-block to $\mathcal X$'s own cone constraints (which hold, since $\tilde\xi_{\mathcal X}\in\mathcal K_{\mathcal X}$)
	and the $\mathcal Y$-block to equality ($-f_j+\tfrac12(d_j+f_j)=\tfrac12(d_j-f_j)$), holding
	regardless of $f_j$. The activation and binary-domain constraints hold since $\delta_{\mathcal X}\in\{-1,1\}^{n_{b,\mathcal X}}\le\mathbf1$
	and $q_{\mathcal Y}=-\mathbf1\le-\mathbf1$. The point map gives
	$x=G_{\mathcal X}\tilde\xi_{\mathcal X}+G_{\mathcal Y}(0,-\mathbf1)+\gamma+c=G_{\mathcal X}\tilde\xi_{\mathcal X}-g_{\mathcal Y}+(c_{\mathcal X}+g_{\mathcal Y})=G_{\mathcal X}\tilde\xi_{\mathcal X}+c_{\mathcal X}=x$,
	using $\gamma+c=c_{\mathcal X}+g_{\mathcal Y}$ (direct expansion) and $G_{\mathcal Y}(0,-\mathbf1)=-g_{\mathcal Y}$.
	The case $x\in\mathcal Y$ is symmetric at $\sigma=-1$, using $-\gamma+c=g_{\mathcal X}+c_{\mathcal Y}$.
	
	\emph{(i), $\mathcal X\cup\mathcal Y\supseteq\mathcal{HC}$:} Let $\xi=(\xi_{\mathcal X},\xi_{\mathcal Y},\sigma)$
	be feasible with $\sigma=1$. The activation constraint $q_{\mathcal Y}\le-\mathbf1$ together
	with $q_{\mathcal Y}\in\{-1,1\}^{n_{b,\mathcal Y}}$ forces $q_{\mathcal Y}=-\mathbf1$.
	Substituting $q_{\mathcal Y}=-\mathbf1$, the equality in $\mathcal Y$ reduces to
	$A_{\mathcal Y}^{\mathrm c}\xi_{\mathcal Y}^{\mathrm c}=0$ and each cone inequality in $\mathcal Y$
	to $\|S_j\xi_{\mathcal Y}^{\mathrm c}\|_{p_j}+(h_j^{\mathrm c})^\top\xi_{\mathcal Y}^{\mathrm c}\le0$, which is the system of Assumption~\ref{ass:bounded} for $\mathcal Y$, so $\xi_{\mathcal Y}^{\mathrm c}=0$. Under Assumption~\ref{ass:bounded}, $\mathcal K_{\mathcal Y}$'s
	recession cone is $\{0\}$, forcing $\xi_{\mathcal Y}^{\mathrm c}=0$. Substituting
	$\xi_{\mathcal X}=(\xi_{\mathcal X}^{\mathrm c},q_{\mathcal X})$ into the block associated with $\mathcal X$,
	the equality row reduces to $A_{\mathcal X}\xi_{\mathcal X}=b_{\mathcal X}$ and each cone
	inequality of $\mathcal X$, so $\xi_{\mathcal X}\in\mathcal K_{\mathcal X}\cap\mathcal B_{\mathcal X}$
	exactly. As above, $x=G_{\mathcal X}\xi_{\mathcal X}-g_{\mathcal Y}+c_{\mathcal X}+g_{\mathcal Y}=G_{\mathcal X}\xi_{\mathcal X}+c_{\mathcal X}\in\mathcal X$.
	The case $\sigma=-1$ is symmetric.
	
	\emph{(ii):} Fix $\sigma\in[-1,1]$, $\alpha,\beta\ge0$, $\alpha+\beta=1$. For $\alpha>0$,
	substitute $q_{\mathcal X}=\alpha r_{\mathcal X}-\beta\mathbf1$, $\xi_{\mathcal X}^{\mathrm c}=\alpha\eta_{\mathcal X}^{\mathrm c}$,
	$r_{\mathcal X}\in[-1,1]^{n_{b,\mathcal X}}$. This substitution reparametrizes $q_{\mathcal X}$ over
	exactly its relaxed range $[-1,\sigma]$, since $r_{\mathcal X}=1\Rightarrow q_{\mathcal X}=\alpha-\beta=\sigma$
	and $r_{\mathcal X}=-1\Rightarrow q_{\mathcal X}=-\alpha-\beta=-1$. Substituting into the
	transformed equality associated with $\mathcal X$ and dividing by $\alpha$ gives
	$A_{\mathcal X}^{\mathrm c}\eta_{\mathcal X}^{\mathrm c}+A_{\mathcal X}^{\mathrm b}r_{\mathcal X}=b_{\mathcal X}$,
	i.e. $\eta_{\mathcal X}:=(\eta_{\mathcal X}^{\mathrm c},r_{\mathcal X})$. Therefore, this point satisfies the equality in $\mathcal X$, the cone inequality reduces to those in $\mathcal X$.
	Hence $\eta_{\mathcal X}$ ranges over $\mathcal K_{\mathcal X}$, with $r_{\mathcal X}\in[-1,1]^{n_{b,\mathcal X}}$
	free, i.e. exactly $\operatorname{relax}(\mathcal K_{\mathcal X}\cap\mathcal B_{\mathcal X})$,
	so $G_{\mathcal X}\eta_{\mathcal X}+c_{\mathcal X}$ ranges over $\operatorname{relax}(\mathcal X)=\operatorname{conv}(\mathcal X)$
	by hypothesis. Moreover $G_{\mathcal X}\xi_{\mathcal X}=\alpha G_{\mathcal X}^{\mathrm c}\eta_{\mathcal X}^{\mathrm c}+G_{\mathcal X}^{\mathrm b}(\alpha r_{\mathcal X}-\beta\mathbf1)=\alpha G_{\mathcal X}\eta_{\mathcal X}-\beta g_{\mathcal X}$.
	The symmetric substitution $q_{\mathcal Y}=\beta r_{\mathcal Y}-\alpha\mathbf1$, $\xi_{\mathcal Y}^{\mathrm c}=\beta\eta_{\mathcal Y}^{\mathrm c}$
	gives $G_{\mathcal Y}\xi_{\mathcal Y}=\beta G_{\mathcal Y}\eta_{\mathcal Y}-\alpha g_{\mathcal Y}$
	with $G_{\mathcal Y}\eta_{\mathcal Y}+c_{\mathcal Y}$ ranging over $\operatorname{conv}(\mathcal Y)$.
	Substituting $\sigma=\alpha-\beta$ into $\gamma\sigma+c$ and collecting terms (direct expansion,
	using $\alpha+\beta=1$) gives $\gamma\sigma+c=\alpha(c_{\mathcal X}+g_{\mathcal Y})+\beta(c_{\mathcal Y}+g_{\mathcal X})$,
	so
	\begin{equation}
		x=G_{\mathcal X}\xi_{\mathcal X}+G_{\mathcal Y}\xi_{\mathcal Y}+\gamma\sigma+c
		=\alpha(G_{\mathcal X}\eta_{\mathcal X}+c_{\mathcal X})+\beta(G_{\mathcal Y}\eta_{\mathcal Y}+c_{\mathcal Y}),
	\end{equation}
	the $\pm\beta g_{\mathcal X}\pm\alpha g_{\mathcal Y}$ terms cancelling exactly against the
	matching terms in $\gamma\sigma+c$. Ranging over $\sigma\in[-1,1]$ (hence all $\alpha,\beta\ge0$,
	$\alpha+\beta=1$) and all $\eta_{\mathcal X},\eta_{\mathcal Y}$ therefore sweeps out exactly
	$\{\alpha x_1+\beta x_2:\alpha,\beta\ge0,\alpha+\beta=1,x_1\in\operatorname{conv}(\mathcal X),x_2\in\operatorname{conv}(\mathcal Y)\}=\operatorname{conv}(\mathcal X\cup\mathcal Y)$.
	The boundary $\alpha=0$ (resp. $\beta=0$) is covered by continuity of this sweep, or directly
	by part (i)'s $\sigma=-1$ (resp. $\sigma=1$) case. \qed
\end{pf}

\begin{rem}
	Assumption~\ref{ass:bounded} is a condition on the generators and not on the set in $\mathbb{R}^n$, and it does not require the set to be nonempty. It holds whenever every continuous generator enters at least one norm, i.e., the matrices $S_i$ jointly select all entries of $\xi^{\mathrm c}$, and $h_i^{\mathrm c}=0$ for all $i$, since the second condition then reads $\|S_i\xi^{\mathrm c}\|_{p_i}\le 0$ for all $i$. In Proposition~\ref{prop:union-base} and Theorem~\ref{thm:union-general}, the linear terms in the cone constraints involve binary generators only, so $h_i^{\mathrm c}$ stays zero. Affine maps, Cartesian products, Minkowski sums and generalized intersections keep every continuous generator inside a norm. The assumption therefore holds for every Hybrid Cone built from zonotopes, constrained zonotopes, ellipsoids or other norm-bounded sets through these operations.
\end{rem}

An important aspect regarding the set descriptions is the generator and constraint growth associated with each set operation. From the statements of Proposition \ref{prop:union-base} and Theorem \ref{thm:union-general} we can state the following corollary.

\begin{cor}
	\label{cor:growthUnion}
	The union of general HCs $\mathcal{X}$ and $\mathcal{Y}$ with continuous generators $n_{gc}^{\mathcal{X}}$ and $n_{gc}^{\mathcal{Y}}$, binary generators $n_{b}^{\mathcal{X}}$ and $n_{b}^{\mathcal{Y}}$, number of constraints $n_{c}^{\mathcal{X}}$ and $n_{c}^{\mathcal{Y}}$, and cone inequality counts $N_{g}^{\mathcal{X}}$ and $N_{g}^{\mathcal{Y}}$ is given by a set $\mathcal{HC}$ with 
	\begin{equation}
		\begin{aligned}
			n_{gc}^{\mathcal{HC}} &= n_{gc}^{\mathcal{X}} + n_{gc}^{\mathcal{Y}}\\
			n_{b}^{\mathcal{HC}} & = n_{b}^{\mathcal{X}} + n_{b}^{\mathcal{Y}} + 1\\
			n_{c}^{\mathcal{HC}} & = n_{c}^{\mathcal{X}} + n_{c}^{\mathcal{Y}}\\
			N_{g}^{\mathcal{HC}} & = N_{g}^{\mathcal{X}} + N_{g}^{\mathcal{Y}} + n_{b}^{\mathcal{X}} + n_{b}^{\mathcal{Y}}.
		\end{aligned}
	\end{equation}
\end{cor}

\begin{pf}
	The corollary follows directly by the construction in Proposition \ref{prop:union-base} where we concatenated the continuous generators from $\mathcal{X}$ and $\mathcal{Y}$, we maintained the same equality constraints but perspectived scaled, a single binary variable was added to represent the union. The number of cone constraints is increased by the number of binary generators in the construction in Theorem \ref{thm:union-general} to generate the perspective scaled inequalities associated with $q_{\mathcal{X}} \leq \sigma \mathbf{1}$ and equivalently for $q_{\mathcal{Y}}$. \qed
\end{pf}

The overhead of the union with respect to its operands is minimal in the following sense.

\begin{prop}[Minimal overhead]
	\label{prop:minimal}
	Let $\mathcal X$ and $\mathcal Y$ be binary-free Hybrid Cones with $\mathcal X\cup\mathcal Y$ nonconvex. Every Hybrid Cone representing $\mathcal X\cup\mathcal Y$ has at least one binary generator. The Hybrid Cone of Proposition~\ref{prop:union-base} has exactly one, and its continuous generators and equality constraints are exactly those of the two operands.
\end{prop}

\begin{pf}
	A Hybrid Cone without binary generators is the affine image of the convex set $\mathcal K$, hence convex. Since $\mathcal X\cup\mathcal Y$ is nonconvex, at least one binary generator is needed. The second claim follows from Corollary~\ref{cor:growthUnion}. \qed
\end{pf}

Therefore, the union carries the generators that describe each operand and the single binary variable that selects between them, and nothing else, in contrast with the union of hybrid zonotopes in \cite{bird:complement}, which adds slack generators and equality constraints. We do not claim that the representation of a particular union cannot be reduced further, for which order reduction can be used.

\section{Reachable Set for ReLU NN}
\label{sec:relu}

In this section, we start by showing that a single ReLU unit admits an exact Hybrid Cone
representation with two continuous generators and one binary generator, while preserving sharpness for the graph set. We then extend it
recursively to a feedforward network of arbitrary depth and width, where exactness is
preserved but sharpness is lost due to the need for a general sharp operation. 

For a scalar pre-activation $v$ and post-activation $x$, let $a>0$ and
\begin{equation}
	\Phi_a := \{(v,x)\in\mathbb R^2 : x=\max(0,v),\ |v|\le a\}
\end{equation}
be the ReLU graph restricted to $[-a,a]$. Intuitively, we can write $\Phi_a$ as the union of an ``off''
branch and an ``on'' branch, each a Hybrid Cone with a single continuous generator and no
binary generators or equality constraints given by
\begin{align}
	\mathcal X_{\mathrm{off}}&=\left(\begin{bmatrix}
		\frac{a}{2} \\0
	\end{bmatrix},\,\begin{bmatrix}
	-\tfrac a2\\0
	\end{bmatrix},\,[\,],\,[\,],\,\{|\lambda_1|\le1\}\right),\\
	\mathcal X_{\mathrm{on}}&=\left(\tfrac a2 \mathbf{1}_{2},\,\tfrac a2 \mathbf{1}_{2},\,[\,],\,[\,],\,\{|\lambda_2|\le1\}\right).
\end{align}

Using this construction, we can present a proposition for the ReLU graph set of a single neuron.

\begin{prop}[ReLU graph of a single neuron]
	\label{prop:relu-single}
	Let $\sigma$ be a scalar binary generator, $\alpha=(1+\sigma)/2$, $\beta=1-\alpha$. The
	Hybrid Cone $\mathcal{HC}_a=(G,c,A,b,\mathfrak G)$ with generator vector
	$\xi=(\lambda_1,\lambda_2,\sigma)$,
	\begin{equation}
		G=\begin{bmatrix} a/2 & a/2 & -a/2\\ 0 & a/2 & -a/4\end{bmatrix},\qquad
		c=\begin{bmatrix}0\\ a/4\end{bmatrix},\quad A=[\,],\ b=[\,],
	\end{equation}
	\begin{equation}
		\mathfrak G=\Bigl\{|\lambda_1|-\tfrac\sigma2\le\tfrac12\Bigr\}\cup\Bigl\{|\lambda_2|+\tfrac\sigma2\le\tfrac12\Bigr\}\cup\bigl\{\sigma\in\{-1,1\}\bigr\},
	\end{equation}
	satisfies $\mathcal{HC}_a=\Phi_a$, with $n_{gc}=2$, $n_b=1$, $n_c=0$, $N_g=2$, and
	$\operatorname{relax}(\mathcal{HC}_a)=\operatorname{conv}(\Phi_a)$.
\end{prop}

\begin{pf}
	$\mathcal X_{\mathrm{off}}=\{(-a/2+\tfrac a2\lambda_1,\,0):|\lambda_1|\le1\}=\{(v,0):v\in[-a,0]\}$
	and $\mathcal X_{\mathrm{on}}=\{(a/2+\tfrac a2\lambda_2,\,a/2+\tfrac a2\lambda_2):|\lambda_2|\le1\}=\{(v,v):v\in[0,a]\}$,
	so $\mathcal X_{\mathrm{off}}\cup\mathcal X_{\mathrm{on}}=\Phi_a$. Both operands are
	binary-free ($n_{b,\mathcal X_{\mathrm{off}}}=n_{b,\mathcal X_{\mathrm{on}}}=0$), so
	Proposition~\ref{prop:union-base} applies with $\mathcal X=\mathcal X_{\mathrm{off}}$,
	$\mathcal Y=\mathcal X_{\mathrm{on}}$. Replacing $A_{\mathcal X}=A_{\mathcal Y}=[\,]$
	into its $A,b$ formulas gives $A=[\,],b=[\,]$. Substituting $c_{\mathcal X}=(-a/2,0)$,
	$c_{\mathcal Y}=(a/2,a/2)$ into its $G,c$ formulas gives exactly the $G,c$ above. Finally, 
	substituting each branch's single constraint ($S_i=1,p_i=\infty,h_i=0,d_i=1$) into its
	$\mathfrak G$ formula gives exactly $\mathfrak G$ above. This gives $\mathcal{HC}_a=\Phi_a$
	by the conclusion (i) of Proposition~\ref{prop:union-base}. Since both operands are binary-free,
	Proposition~\ref{prop:union-base} in its (ii) applies in the same step, giving
	$\operatorname{relax}(\mathcal{HC}_a)=\operatorname{conv}(\Phi_a)$, i.e., the single-neuron
	graph is both exact and sharp. Applying Corollary~\ref{cor:growthUnion} with
	$n_{gc}^{\mathcal X}=n_{gc}^{\mathcal Y}=1$, $n_b^{\mathcal X}=n_b^{\mathcal Y}=0$,
	$n_c^{\mathcal X}=n_c^{\mathcal Y}=0$, $N_g^{\mathcal X}=N_g^{\mathcal Y}=1$ gives
	$n_{gc}=2$, $n_b=1$, $n_c=0$, $N_g=2$. \qed
\end{pf}

This is half the continuous generators, and no equality constraints at all, compared to
the exact HZ representation of a ReLU unit in \cite{ortiz:reluHZ}, which is built as the
intersection of three zonotopes and costs four continuous generators, one binary
generator, and two equality constraints per neuron. The saving comes from building
$\Phi_a$ directly as a sharp union of its two branches rather than as an intersection of
enclosing zonotopes that must then be cut back down to the graph by equality
constraints. Unlike the big-$M$ encodings used for exact ReLU verification in
\cite{fischetti:relu}, no auxiliary disjunctive constant is introduced here since 
the bound $a$ enters $G$ and $\mathfrak G$ homogeneously, as the same scale factor that
already fixes the branch endpoints, rather than as a separate big-$M$ slack.

A feedforward ReLU network of depth $L$ is the map $\phi:\mathbb R^{n_0}\to\mathbb R^{n_L}$
given layerwise by $x_l=\max(0,W_lx_{l-1}+b_l)$, $l=1,\dots,L$, with weights
$W_l\in\mathbb R^{n_l\times n_{l-1}}$ and biases $b_l\in\mathbb R^{n_l}$. For an input
domain $\mathcal X_0\subset\mathbb R^{n_0}$, itself a Hybrid Cone, we want the exact graph
\begin{equation}
	\Gamma_\phi := \{(x_0,x_L)\in\mathbb R^{n_0+n_L} : x_L=\phi(x_0),\ x_0\in\mathcal X_0\}.
\end{equation}
This is an instance of the graph-composition principle illustrated in Section~\ref{sec:problem-nonlinear} where each layer contributes one more generalized intersection and projection, per neuron, to the running graph set, in the same way
functional decomposition composes low-dimensional node graphs.

For $l=1,\dots,L$, define the block-extended weight and bias
\begin{align}
	\widehat W_l &:= \begin{bmatrix} I_{n_0} & 0\\ 0 & W_l\end{bmatrix}\in\mathbb R^{(n_0+n_l)\times(n_0+n_{l-1})},
	\\
	\widehat b_l &:= \begin{bmatrix}0_{n_0}\\ b_l\end{bmatrix}\in\mathbb R^{n_0+n_l},
\end{align}
which act as identity on the $x_0$-block and as $(W_l,b_l)$ on the $x_{l-1}$-block. For $a\in\mathbb R^n_{>0}$ write $\operatorname{reluLayer}(a):=\Phi_{a_1}\times\cdots\times\Phi_{a_n}$, with its coordinates ordered as $(v,x)\in\mathbb R^n\times\mathbb R^n$, i.e., all pre-activations first. We identify $\mathcal S_l$ with the graph pair $\{(x_0,x_l):x_l=\phi_{1:l}(x_0),\,x_0\in\mathcal X_0\}\subset\mathbb R^{n_0+n_l}$, with $\mathcal S_0=\mathcal X_0$ embedded as $x_0\mapsto(x_0,x_0)$. Each layer is built in three steps. First, $\widehat W_l\mathcal S_{l-1}+\widehat b_l$ holds the pairs $(x_0,v_l)$. Second, the set $\mathcal{U}_{l}$, the Cartesian product of this set with $\operatorname{reluLayer}(a^{(l)})$, has coordinates $(x_0,v,v',x)\in\mathbb R^{n_0+3n_l}$, and a generalized intersection with the singleton $\{0\}\subset\mathbb R^{n_l}$ under
\begin{equation}
	R_l=\begin{bmatrix}0 & I_{n_l} & -I_{n_l} & 0\end{bmatrix}
\end{equation}
enforces $v=v'$. Third, the linear map
\begin{equation}
	H_l=\begin{bmatrix}I_{n_0}&0&0&0\\0&0&0&I_{n_l}\end{bmatrix}
\end{equation}
keeps only $(x_0,x)$. Only the second step adds constraints, namely $n_l$ equalities.

\begin{thm}[ReLU network graph]
	\label{thm:relu-network}
	Let $\mathcal S_0=\begin{bmatrix}I_{n_0}\\I_{n_0}\end{bmatrix}\mathcal X_0$ and, for $l=1,\dots,L$,
	\begin{equation}
		\mathcal U_l=\bigl(\widehat W_l\mathcal S_{l-1}+\widehat b_l\bigr)\times\operatorname{reluLayer}(a^{(l)}),
		\quad
		\mathcal S_l=H_l\bigl(\mathcal U_l\cap_{R_l}\{0\}\bigr).
		\label{eq:Sl}
	\end{equation}
	If
	\begin{equation}
		\bigl|\bigl(W_l\phi_{1:l-1}(x_0)+b_l\bigr)_i\bigr|\le a_i^{(l)}\quad\text{for all }x_0\in\mathcal X_0,
		\label{eq:relu-bound}
	\end{equation}
	holds for every layer $l=1,\dots,L$ and every neuron $i=1,\dots,n_l$, then $\mathcal S_L=\Gamma_\phi$.
\end{thm}

\begin{pf}
	We proceed by induction on $l$, the base case holding by construction. Assume $\mathcal S_{l-1}=\{(x_0,\phi_{1:l-1}(x_0)):x_0\in\mathcal X_0\}$. Since $\widehat W_l$ is the identity on the $x_0$-block, $\widehat W_l\mathcal S_{l-1}+\widehat b_l=\{(x_0,v):v=W_l\phi_{1:l-1}(x_0)+b_l,\,x_0\in\mathcal X_0\}$. The set $\mathcal{U}_{l}$ adds points $(v',x)$ with $|v'|\le a^{(l)}$ and $x=\max(0,v')$ entrywise by Proposition~\ref{prop:relu-single}, and the intersection with $\{0\}$ under $R_l$ keeps only $v=v'$. By \eqref{eq:relu-bound}, the bound $|v'|\le a^{(l)}$ holds at every such point, so none is removed, and $x=\max(0,W_l\phi_{1:l-1}(x_0)+b_l)=\phi_{1:l}(x_0)$. The map $H_l$ discards $(v,v')$, hence $\mathcal S_l=\{(x_0,\phi_{1:l}(x_0)):x_0\in\mathcal X_0\}$, and taking $l=L$ gives $\mathcal S_L=\Gamma_\phi$. \qed
\end{pf}

If the network ends with an affine layer $y=W_{L+1}x_L+b_{L+1}$ without activation, it is applied to $\mathcal S_L$ as the affine map of Proposition~\ref{prop:ccgOps}, which changes the output map and adds no generators, equalities or cone constraints. Therefore, the counts of Corollary~\ref{cor:relu-growth} are unchanged by it.

The proof above fixes $\mathcal S_l$ as a set, but the growth in the number of generators and constraints is read off from its tuple, so we provide the explicit format for the structure. Let $\mathcal S_{l-1}=(G_{l-1},c_{l-1},A_{l-1},b_{l-1},\mathfrak G_{l-1})$ and $\operatorname{reluLayer}(a^{(l)})=(G_\Phi,c_\Phi,[\,],[\,],\mathfrak G_\Phi)$, which has no equalities because the per-neuron sets of Proposition~\ref{prop:relu-single} have none. Let $\Pi_{0}=[I_{n_0}\ 0]$ and $\Pi_{l-1}=[0\ I_{n_{l-1}}]$ select $x_0$ and $x_{l-1}$ from a point of $\mathcal S_{l-1}$, and let $\Pi_v=[I_{n_l}\ 0]$ and $\Pi_x=[0\ I_{n_l}]$ select the pre-activations $v$ and the activations $x$ from a point of $\operatorname{reluLayer}(a^{(l)})$. Then $\mathcal S_l=(\bar G_l,\bar c_l,\bar A_l,\bar b_l,\bar{\mathfrak G}_l)$ has generator vector $(\xi_{l-1},\xi_\Phi)$ and
\begin{equation}
	\bar G_l=\begin{bmatrix}\Pi_{0}G_{l-1}&0\\0&\Pi_xG_\Phi\end{bmatrix},\qquad
	\bar c_l=\begin{bmatrix}\Pi_{0}c_{l-1}\\\Pi_xc_\Phi\end{bmatrix},
\end{equation}
\begin{equation}
	\resizebox{\columnwidth}{!}{$
	\bar A_l=\begin{bmatrix}A_{l-1}&0\\W_l\Pi_{l-1}G_{l-1}&-\Pi_vG_\Phi\end{bmatrix},\
	\bar b_l=\begin{bmatrix}b_{l-1}\\\Pi_vc_\Phi-W_l\Pi_{l-1}c_{l-1}-b_l\end{bmatrix}$},
\end{equation}
and $\bar{\mathfrak G}_l=\mathfrak G_{l-1}\cup\mathfrak G_\Phi$. The output map reads $x_0$ from the generators of $\mathcal S_{l-1}$ and $x_l$ from those of the ReLU layer. The first row of equalities keeps the constraints of $\mathcal S_{l-1}$. The second row imposes that the pre-activation $W_lx_{l-1}+b_l$ computed from $\mathcal S_{l-1}$ equals the pre-activation of the ReLU layer. Hence each layer appends the generators and cone constraints of $\operatorname{reluLayer}(a^{(l)})$ and exactly $n_l$ equalities, which is the growth counted in Corollary~\ref{cor:relu-growth}.

\begin{rem}[Exact, but not sharp]
	Theorem~\ref{thm:relu-network} establishes exactness of $\mathcal S_L$, not sharpness of
	its relaxation. Each layer's $\Phi_l$ is itself sharp, being a Cartesian product of the sharp
	per-neuron sets of Proposition~\ref{prop:relu-single}. This is due to
	$\operatorname{conv}(\mathcal{X}\times \mathcal{Y})=\operatorname{conv}(\mathcal{X})\times\operatorname{conv}(\mathcal{Y})$ for
	arbitrary sets, and, since relaxation acts entrywise on a block-diagonal construction,
	$\operatorname{relax}(\Phi_l)=\operatorname{conv}(\Phi_l)$. What breaks the guarantee is
	the generalized intersection used to form $\mathcal S_l$ from $\mathcal S_{l-1}$ and
	$\Phi_l$. Unlike union, generalized intersection does not in general preserve sharpness since it can couple binary
	selectors from $\mathcal S_{l-1}$ and $\Phi_l$ in a way that cuts across branches rather
	than combining whole branches.
\end{rem}

\begin{cor}
	\label{cor:relu-growth}
	Let $N:=\sum_{l=1}^L n_l$ be the total number of ReLU neurons in the network. Then
	\begin{equation}
		\label{eq:complexityNN}
		\begin{aligned}
			n_{gc}^{\Gamma_\phi} &= n_{gc}^{\mathcal X_0} + 2N, &
			n_b^{\Gamma_\phi} &= n_b^{\mathcal X_0} + N,\\
			n_c^{\Gamma_\phi} &= n_c^{\mathcal X_0} + N, &
			N_g^{\Gamma_\phi} &= N_g^{\mathcal X_0} + 2N.
		\end{aligned}
	\end{equation}
\end{cor}

\begin{pf}
	By Proposition~\ref{prop:relu-single} and the Cartesian product's additive growth, 
	$\Phi_l$ contributes $2n_l$ continuous generators, $n_l$
	binary generators, $0$ equality constraints, and $2n_l$ cone inequalities. The
	generalized intersection at layer $l$ adds one coupling equality per neuron ($v_{l,i}$
	tied to $\Phi_l$'s $i$-th pre-activation row) and no new cone inequalities, since
	generalized intersection does not introduce cone constraints of its own. Summing these per-layer increments over
	$l=1,\dots,L$ and adding the base counts of $\mathcal X_0$ gives the increase in \eqref{eq:complexityNN}. \qed
\end{pf}

\section{Complement of Hybrid Cones}
\label{sec:complement}

The admissible set in obstacle avoidance and robust safety problems is the complement of a union of obstacles within a bounded region, so a set representation closed under this operation avoids converting to hyperplane arrangements. Hybrid zonotopes have this property as shown in \cite{bird:complement} that the complement is the intersection of the complements of the constrained-zonotope leaves, and the complement of each leaf is encoded through the KKT conditions of a linear program. The cost is a representation that grows with the number of leaves and that carries slack generators whose only purpose is to turn inequalities into equalities. In this section we show that Hybrid Cones remove both sources of growth. Native linear inequalities eliminate the slack generators, and when a Hybrid Cone is built from unions, its leaves are read off the union structure instead of being enumerated from the $2^{n_b}$ binary assignments.

We restrict attention to Hybrid Cones whose continuous cones are $\ell_\infty$ boxes, which is the setting of hybrid zonotopes. This restriction belongs to the operation and not to the representation since the complement of an $\ell_\infty$ ball has finitely many facets, so the complementary slackness conditions reduce to finitely many binary decisions, whereas for a smooth norm such as $p=2$ the boundary of the complement is described by a continuum of tangent hyperplanes and no finite binary encoding exists.

For a set $\mathcal S\subset\mathbb R^n$ and a bounded convex region of interest $\mathcal D\subset\mathbb R^n$, a box in practice, we denote by
\begin{equation}
	\mathcal C_{\mathcal D}(\mathcal S):=\bigl(\mathbb R^n\setminus\operatorname{int}\mathcal S\bigr)\cap\mathcal D
\end{equation}
the closure of the complement of $\mathcal S$ within $\mathcal D$. The bounded region is required because complementary slackness is encoded with big-$M$ constants. A binary-free Hybrid Cone with a single $\ell_\infty$ cone is composed of a single leaf corresponding to a constrained zonotope $\mathcal Z=(G,c,A,b)$. In case the Hybrid Cone has multiple leaves, they are assumed nonempty and full dimensional, with $[G^\top\ A^\top]$ of full column rank. Lower dimensional leaves have empty interior and can be discarded, since they do not change $\mathcal C_{\mathcal D}$.

For a leaf $\mathcal Z\subseteq\mathcal D$ we use scalars $\bar\delta\ge0$ and $\bar\lambda>0$ such that
\begin{align}
	&\mathcal D\subseteq\{G\xi+c:\ A\xi=b,\ \|\xi\|_\infty\le1+\bar\delta\},\label{eq:cover}\\
	&\bar\lambda\ge\max\{\|\lambda\|_\infty:\ \|[G^\top\ A^\top]\lambda\|_\infty\le1\},\label{eq:dualbound}
\end{align}
and write $\kappa:=1+\bar\delta$. We discuss how to compute them after the next result. The complement of a single leaf is the building block of everything that follows, so we present it first, with the generators, constraints and cones written out explicitly.

\begin{prop}[Complement of a leaf]
	\label{prop:comp-leaf}
	Let $\mathcal Z=(G,c,A,b)\subset\mathbb R^n$ be a leaf with $n_{gc}$ continuous generators and $n_c$ equality constraints, and let $\mathcal D\supseteq\mathcal Z$ and $\bar\delta,\bar\lambda$ satisfy \eqref{eq:cover}--\eqref{eq:dualbound}. Define the Hybrid Cone $\widetilde{\mathcal Z}=(\tilde G,\tilde c,\tilde A,\tilde b,\tilde{\mathfrak G})$ with continuous generators $\zeta=(\xi,\eta,\lambda,\mu_1,\mu_2)\in\mathbb R^{3n_{gc}+n+n_c+1}$ and binary generators $\pi=(\pi_1,\pi_2)\in\{-1,1\}^{2n_{gc}}$, where $\xi,\mu_1,\mu_2,\pi_1,\pi_2\in\mathbb R^{n_{gc}}$, $\eta\in\mathbb R$ and $\lambda\in\mathbb R^{n+n_c}$, by
	\begin{equation}
		\tilde G=\begin{bmatrix}\kappa G&0\end{bmatrix},\qquad \tilde c=c,
	\end{equation}
	\begin{equation}
		\tilde A=\scalebox{0.8}{$\begin{bmatrix}
				\kappa A&0&0&0&0&0\\
				0&0&\bar\lambda[G^\top\ A^\top]&\tfrac12 I&-\tfrac12 I&0\\
				0&0&0&\tfrac12\mathbf 1^\top&\tfrac12\mathbf 1^\top&0
			\end{bmatrix}$},\qquad
		\tilde b=\begin{bmatrix}b\\0\\1-n_{gc}\end{bmatrix},
	\end{equation}
	where the columns of $\tilde A$ are ordered as $(\xi,\eta,\lambda,\mu_1,\mu_2,\pi)$, and with cone constraints $\tilde{\mathfrak G}$ consisting of the single $\ell_\infty$ cone $\|\zeta\|_\infty\le1$ on the continuous generators, the binary domain $\pi\in\{-1,1\}^{2n_{gc}}$, and the six families of linear inequalities
	\begin{equation}
		\begin{aligned}
			&\kappa\xi-\tfrac{\bar\delta}2\eta\mathbf 1\le\bigl(1+\tfrac{\bar\delta}2\bigr)\mathbf 1, &&
			-\kappa\xi-\tfrac{\bar\delta}2\eta\mathbf 1\le\bigl(1+\tfrac{\bar\delta}2\bigr)\mathbf 1,\\
			&-\kappa\xi+\tfrac{\bar\delta}2\eta\mathbf 1+\kappa\pi_1\le\tfrac{\bar\delta}2\mathbf 1, &&
			\ \ \kappa\xi+\tfrac{\bar\delta}2\eta\mathbf 1+\kappa\pi_2\le\tfrac{\bar\delta}2\mathbf 1,\\
			&\mu_1-\pi_1\le0, &&\mu_2-\pi_2\le0.
		\end{aligned}
		\label{eq:comp-ineq}
	\end{equation}
	Then $\widetilde{\mathcal Z}\cap\mathcal D=\mathcal C_{\mathcal D}(\mathcal Z)$, and $\widetilde{\mathcal Z}$ has
	\begin{align}
		\tilde n_{gc}&=3n_{gc}+n+n_c+1,\quad
		\tilde n_b=2n_{gc},\quad\\
		\tilde n_c&=n_c+n_{gc}+1,\quad
		\tilde N_g=6n_{gc}+1.
	\end{align}
\end{prop}

\begin{pf}
	For $x\in\mathbb R^n$ let $\delta^\star(x)$ be the optimal value of
	\begin{equation}
		\min_{\hat\xi,\delta}\ \delta\quad\text{s.t.}\quad G\hat\xi+c=x,\ A\hat\xi=b,\ \|\hat\xi\|_\infty\le1+\delta,
		\label{eq:lp}
	\end{equation}
	which is feasible for every $x$ because $[G^\top\ A^\top]$ has full column rank. By Lemma~1 of \cite{bird:complement}, $x\in\overline{\mathbb R^n\setminus\mathcal Z}$ if and only if $\delta^\star(x)\ge0$, and \eqref{eq:cover} gives $\delta^\star(x)\le\bar\delta$ for every $x\in\mathcal D$.
	
	We can build the KKT system for the complement by letting $\hat\lambda$ be the multiplier of the equalities and $\hat\mu_1,\hat\mu_2\ge0$ those of $\hat\xi\le(1+\delta)\mathbf 1$ and $-\hat\xi\le(1+\delta)\mathbf 1$. Since \eqref{eq:lp} is a linear program, $\delta=\delta^\star(x)$ if and only if
	\begin{align}
		&G\hat\xi+c=x,\quad A\hat\xi=b,\quad\|\hat\xi\|_\infty\le1+\delta,\label{eq:kkt-p}\\
		&[G^\top\ A^\top]\hat\lambda+\hat\mu_1-\hat\mu_2=0,\quad\mathbf 1^\top(\hat\mu_1+\hat\mu_2)=1,\quad\hat\mu_1,\hat\mu_2\ge0,\label{eq:kkt-d}\\
		&\hat\mu_{1,i}(\hat\xi_i-\delta-1)=0,\quad\hat\mu_{2,i}(-\hat\xi_i-\delta-1)=0,\quad 1 \leq i \leq n_{gc},\label{eq:kkt-cs}
	\end{align}
	hold for some multipliers.
	
	Equation \eqref{eq:kkt-d} gives $\hat\mu_j\in[0,1]^{n_{gc}}$ and $\|[G^\top\ A^\top]\hat\lambda\|_\infty\le1$, hence $\|\hat\lambda\|_\infty\le\bar\lambda$. For $x\in\mathcal D$ with $\delta=\delta^\star(x)\ge0$ we have $\delta\le\bar\delta$ and $\|\hat\xi\|_\infty\le\kappa$, so $\hat\xi_i-\delta-1\ge-2\kappa$. Therefore \eqref{eq:kkt-cs} holds if and only if there are $p_1,p_2\in\{0,1\}^{n_{gc}}$ with $\hat\mu_j\le p_j$ and
	\begin{equation}
		-2\kappa(\mathbf 1-p_1)\le\hat\xi-\delta\mathbf 1-\mathbf 1,\qquad-2\kappa(\mathbf 1-p_2)\le-\hat\xi-\delta\mathbf 1-\mathbf 1.
	\end{equation}
	Indeed, $p_{j,i}=1$ makes the second inequality tight against \eqref{eq:kkt-p}, $p_{j,i}=0$ forces $\hat\mu_{j,i}=0$ and relaxes the inequality, and conversely we choose $p_{j,i}=1$ if and only if $\hat\mu_{j,i}>0$.
	
	We substitute $\hat\xi=\kappa\xi$, $\delta=\tfrac{\bar\delta}2(\eta+1)$, $\hat\lambda=\bar\lambda\lambda$, $\hat\mu_j=\tfrac12(\mu_j+1)$ and $p_j=\tfrac12(\pi_j+1)$, so that all continuous generators lie in $[-1,1]$ and all binary generators in $\{-1,1\}$. Then $x=\kappa G\xi+c$ and $\kappa A\xi=b$ give the output map and the first block of equalities. Dual feasibility becomes $\bar\lambda[G^\top\ A^\top]\lambda+\tfrac12\mu_1-\tfrac12\mu_2=0$ and $\tfrac12\mathbf 1^\top(\mu_1+\mu_2)=1-n_{gc}$, the second and third blocks. The primal inequality $\|\hat\xi\|_\infty\le1+\delta$ becomes the first two families in \eqref{eq:comp-ineq}. The big-$M$ inequalities become the third and fourth families, using $2\kappa(1-p_j)=\kappa(1-\pi_j)$ and $\kappa-1=\bar\delta$, and $\hat\mu_j\le p_j$ becomes the last two.
	
	If $x\in\widetilde{\mathcal Z}\cap\mathcal D$, the generators satisfy \eqref{eq:kkt-p}--\eqref{eq:kkt-cs} with $\delta\in[0,\bar\delta]$, hence $\delta=\delta^\star(x)\ge0$ and $x\in\mathcal C_{\mathcal D}(\mathcal Z)$. Conversely, if $x\in\mathcal C_{\mathcal D}(\mathcal Z)$, then $\delta^\star(x)\in[0,\bar\delta]$ and an optimal primal-dual pair exists with $\|\hat\lambda\|_\infty\le\bar\lambda$ and $\|\hat\xi\|_\infty\le\kappa$. Taking $p_j$ as above and reversing the normalization gives generators in the admissible domains, so $x\in\widetilde{\mathcal Z}$. The counts for generators and constraints follow by inspection of the dimensions. \qed
\end{pf}

The system \eqref{eq:comp-ineq} is exactly where hybrid zonotopes pay for the lack of inequalities. Adding to each of the $6n_{gc}$ rows a slack generator in $[-1,1]$, whose coefficient is half the width of the range of the row over the admissible box, turns \eqref{eq:comp-ineq} into equalities and recovers the identity of \cite[Proposition~2]{bird:complement}, with $9n_{gc}+n+n_c+1$ continuous generators, $2n_{gc}$ binary generators and $7n_{gc}+n_c+1$ equalities. The Hybrid Cone version therefore has $6n_{gc}$ fewer continuous generators and $6n_{gc}$ fewer equality constraints, and the slack-free inequalities are carried by the cone definition.

\begin{rem}[Computing $\bar\delta$ and $\bar\lambda$]
	The smallest admissible $\bar\delta$ is $\max_{x\in\mathcal D}\min\{\|\xi\|_\infty:\ G\xi+c=x,\ A\xi=b\}-1$. This is the maximum of a convex piecewise-linear function over $\mathcal D$, so it is attained at a vertex of $\mathcal D$, which requires $2^n$ linear programs with $n_{gc}+1$ variables for a box. Any larger value is admissible at the price of a larger big-$M$. The minimization over $\xi$ must be taken inside: maximizing $\|\xi\|_\infty$ over all $\xi$ with $A\xi=b$ and $G\xi+c\in\mathcal D$ is unbounded as soon as $n_{gc}>n+n_c$, because $[G^\top\ A^\top]^\top$ then has a nontrivial kernel. The constant $\bar\lambda$ is obtained from $2(n+n_c)$ linear programs, which are bounded because $[G^\top\ A^\top]$ has full column rank.
\end{rem}

A Hybrid Cone with several leaves is complemented leaf by leaf. Given leaves $\mathcal Z_1,\dots,\mathcal Z_L$ with centers $c_\tau$ and constants $\bar\delta_\tau,\bar\lambda_\tau$, let $\widetilde{\mathcal Z}_\tau=(\tilde G_\tau,c_\tau,\tilde A_\tau,\tilde b_\tau,\tilde{\mathfrak G}_\tau)$ be given by Proposition~\ref{prop:comp-leaf}. The \emph{leafwise complement} $\widetilde{\mathcal{HC}}$ has generator vector $(\zeta_1,\pi_1,\dots,\zeta_L,\pi_L)$, output map $\tilde G=[\tilde G_1\ 0\ \cdots\ 0]$, $\tilde c=c_1$, the equalities of every $\widetilde{\mathcal Z}_\tau$ stacked block-diagonally, the cone constraints of all $\widetilde{\mathcal Z}_\tau$ with their $\ell_\infty$ cones merged into one, and the coupling equalities
\begin{equation}
	\tilde G_1\begin{bmatrix}\zeta_1\\\pi_1\end{bmatrix}-\tilde G_\tau\begin{bmatrix}\zeta_\tau\\\pi_\tau\end{bmatrix}=c_\tau-c_1,\qquad\tau=2,\dots,L,
	\label{eq:leafwise}
\end{equation}
which enforce that all leaf complements describe the same point $x$ and add $(L-1)n$ equalities. The next assumption states when the intersection of leaf complements is the complement of the union.

\begin{assum}
	\label{ass:interior}
	The leaves satisfy $\operatorname{int}\bigcup_\tau\mathcal Z_\tau=\bigcup_\tau\operatorname{int}\mathcal Z_\tau$.
\end{assum}

We adopt Assumption~\ref{ass:interior} to follow the same setup as \cite{bird:complement}, where the complement of a union is also computed as the intersection of the complements of its leaves, so that the sizes of the two data structures are compared on the same footing. The assumption holds trivially for one leaf, for leaves that are pairwise separated, and for leaves that overlap on a set with nonempty interior. It fails when two leaves touch along a boundary without overlapping, for example two unit squares sharing an edge, because the shared edge lies in the interior of the union but in the interior of neither square. In that case $\bigcap_\tau\mathcal C_{\mathcal D}(\mathcal Z_\tau)$ contains the shared interfaces and is an outer approximation of $\mathcal C_{\mathcal D}(\bigcup_\tau\mathcal Z_\tau)$ that differs from it only there, which is also the behavior of the De Morgan step in \cite[Theorem~1]{bird:complement}.

We now state the complement for a Hybrid Cone that has the structure of a hybrid zonotope, namely a single $\ell_\infty$ cone on the continuous generators and arbitrary equality coupling to the binary generators. No other structure is available, so the leaves must be obtained by fixing the binary generators.

\begin{thm}[Complement of a hybrid zonotope in Hybrid Cone format]
	\label{thm:comp-hz}
	Let $\mathcal{HC}=(G,c,A,b,\mathfrak G)$ with $\mathfrak G=\{\|\xi^{\mathrm c}\|_\infty\le1\}$ and $n_b$ binary generators $\xi^{\mathrm b}$, write $G=[G^{\mathrm c}\ G^{\mathrm b}]$, $A=[A^{\mathrm c}\ A^{\mathrm b}]$, and for $\tau\in\{-1,1\}^{n_b}$ let $\mathcal Z_\tau:=(G^{\mathrm c},\,c+G^{\mathrm b}\tau,\,A^{\mathrm c},\,b-A^{\mathrm b}\tau)$. Let $\mathcal T$ be the set of $\tau$ for which $\mathcal Z_\tau$ is nonempty and full dimensional, $L:=|\mathcal T|\le2^{n_b}$, and let $\mathcal D\supseteq\mathcal{HC}$ be a bounded convex region. Under Assumption~\ref{ass:interior}, $\mathcal C_{\mathcal D}(\mathcal{HC})=\widetilde{\mathcal{HC}}\cap\mathcal D$, where $\widetilde{\mathcal{HC}}$ is the leafwise complement of $\{\mathcal Z_\tau\}_{\tau\in\mathcal T}$, with
	\begin{equation}
		\begin{aligned}
			\tilde n_{gc}&=L(3n_{gc}+n_c+n+1), & \tilde n_b&=2Ln_{gc},\\
			\tilde n_c&=L(n_{gc}+n_c+1)+(L-1)n, & \tilde N_g&=6Ln_{gc}+1.
		\end{aligned}
	\end{equation}
\end{thm}

\begin{pf}
	Fixing $\xi^{\mathrm b}=\tau$ in the definition of $\mathcal{HC}$ gives $\mathcal{HC}=\bigcup_{\tau}\mathcal Z_\tau$, and empty or lower dimensional leaves do not contribute to $\operatorname{int}\mathcal{HC}$. Under Assumption~\ref{ass:interior}, $\mathcal C_{\mathcal D}(\mathcal{HC})=\mathcal D\setminus\bigcup_{\tau\in\mathcal T}\operatorname{int}\mathcal Z_\tau=\bigcap_{\tau\in\mathcal T}\mathcal C_{\mathcal D}(\mathcal Z_\tau)$. By Proposition~\ref{prop:comp-leaf} each factor equals $\widetilde{\mathcal Z}_\tau\cap\mathcal D$, and \eqref{eq:leafwise} implements the intersection, which gives the first claim. Every leaf has $n_{gc}$ continuous generators and $n_c$ equalities, so summing the counts of Proposition~\ref{prop:comp-leaf} over $L$ leaves, adding the $(L-1)n$ coupling equalities, and counting the merged $\ell_\infty$ cone once gives the stated values. \qed
\end{pf}

Theorem~\ref{thm:comp-hz} is as good as the format allows: the leaves must be enumerated through the binary assignments, which takes up to $2^{n_b}$ feasibility problems, and each leaf carries all $n_{gc}$ continuous generators of the hybrid zonotope. A Hybrid Cone built from unions has more structure to exploit, since its leaves are the operands of the unions. The following lemma makes this precise.

\begin{lem}[Leaf recovery]
	\label{lem:leaf-recovery}
	Let $\mathcal{HC}$ be the union of $\mathcal{HC}_{\mathcal X}$ and $\mathcal{HC}_{\mathcal Y}$ given by Theorem~\ref{thm:union-general}, with selector $\sigma$. Fixing $\sigma=1$ and $q_{\mathcal Y}=-\mathbf 1$ and deleting the columns of $\xi_{\mathcal Y}$ together with the constraints that involve only deleted generators yields exactly the tuple of $\mathcal{HC}_{\mathcal X}$. Fixing $\sigma=-1$ and $q_{\mathcal X}=-\mathbf 1$ yields the tuple of $\mathcal{HC}_{\mathcal Y}$ in the same way.
\end{lem}

\begin{pf}
	This is the content of the $\subseteq$ part of the proof of Theorem~\ref{thm:union-general}. At $\sigma=1$ the activation constraints force $q_{\mathcal Y}=-\mathbf 1$, Assumption~\ref{ass:bounded} forces $\xi_{\mathcal Y}^{\mathrm c}=0$, the $\mathcal Y$-rows of $A$ and $\mathfrak G$ reduce to identities, the $\mathcal X$-rows reduce to $A_{\mathcal X}\xi_{\mathcal X}=b_{\mathcal X}$ and $\mathfrak G_{\mathcal X}$, and the output map reduces to $G_{\mathcal X}\xi_{\mathcal X}+c_{\mathcal X}$. The case $\sigma=-1$ is symmetric. \qed
\end{pf}

Applying Lemma~\ref{lem:leaf-recovery} recursively along the union tree yields all leaves with $L-1$ substitutions, with no optimization and no enumeration of binary assignments, and every recovered leaf has only the generators and equalities of the corresponding operand.

\begin{thm}[Complement of a union-generated Hybrid Cone]
	\label{thm:comp-union}
	Let $\mathcal{HC}$ be obtained from $L$ leaves $\mathcal Z_1,\dots,\mathcal Z_L$ by $L-1$ unions as in Theorem~\ref{thm:union-general}, in any nesting order, let $\mathcal D\supseteq\mathcal{HC}$ be a bounded convex region, and let Assumption~\ref{ass:interior} hold. Then $\mathcal C_{\mathcal D}(\mathcal{HC})=\widetilde{\mathcal{HC}}\cap\mathcal D$, where $\widetilde{\mathcal{HC}}$ is the leafwise complement of $\{\mathcal Z_\tau\}_{\tau=1}^{L}$, and in terms of the counts of $\mathcal{HC}$,
	\begin{equation}
		\begin{aligned}
			\tilde n_{gc}&=3n_{gc}+n_c+(n_b+1)(n+1), & \tilde n_b&=2n_{gc},\\
			\tilde n_c&=n_{gc}+n_c+(n_b+1)+n_bn, & \tilde N_g&=6n_{gc}+1.
		\end{aligned}
	\end{equation}
\end{thm}

\begin{pf}
	By Theorem~\ref{thm:union-general}, $\mathcal{HC}=\bigcup_\tau\mathcal Z_\tau$, and Lemma~\ref{lem:leaf-recovery} recovers the leaves. The argument of Theorem~\ref{thm:comp-hz} then gives $\mathcal C_{\mathcal D}(\mathcal{HC})=\widetilde{\mathcal{HC}}\cap\mathcal D$. By Corollary~\ref{cor:growthUnion}, the continuous generators and equality constraints add under union and each union adds one binary generator, hence $\sum_\tau n_{gc}^{\tau}=n_{gc}$, $\sum_\tau n_c^{\tau}=n_c$ and $L=n_b+1$. Summing Proposition~\ref{prop:comp-leaf} over the leaves and adding the $(L-1)n$ coupling equalities gives the stated counts. \qed
\end{pf}

The counts of Theorem~\ref{thm:comp-union} are linear in the size of the input, and the number of leaves enters only through $n_b+1$ and through the coupling equalities. We now compare them with the hybrid zonotope route of \cite{bird:complement}.

\begin{cor}[Order comparison]
	\label{cor:comp-growth}
	Let $\mathcal Z_1,\dots,\mathcal Z_L$ have $g$ generators and $k$ equalities each. The application of Theorem~\ref{thm:comp-union} to the union-generated Hybrid Cone corresponding to $\cup_{i = 1}^{L} \mathcal{Z}_{i}$ will have data structure complexity given by
	\begin{align}
		\tilde n_{gc}&=L(3g+k+n+1),\quad\tilde n_b=2Lg,\\
		\tilde n_c&=L(g+k+1)+(L-1)n,
	\end{align}
	whereas the union of the same leaves written as a hybrid zonotope by $L-1$ applications of \cite[Proposition~1]{bird:complement} from left to right has $g_L=(L^2+2L-2)g+(L-1)(L-2)$ continuous generators and $k_L=Lk+(L^2+L-2)g+(L-1)(L-2)$ equalities, and its complement through \cite[Theorem~1]{bird:complement} has
	\begin{align}
		\tilde n_{gc}&=L(9g_L+n+k_L+1),\quad\tilde n_b=2Lg_L,\quad \\
		\tilde n_c&=L(7g_L+k_L+1)+(L-1)n.
	\end{align}
	The first set of counts is linear in $L$, the second is cubic in $L$.
\end{cor}

\begin{pf}
	The Hybrid Cone counts are Theorem~\ref{thm:comp-union} with $n_{gc}=Lg$, $n_c=Lk$ and $n_b=L-1$. For the hybrid zonotope route, applying \eqref{eq:comp-growth-hz} below, namely the union identity of \cite[Proposition~1]{bird:complement} with growth
	\begin{align}
		n_g&=n_{g,1}+n_{g,2}+2(n_{g,1}^{r}+n_{b,1}^{r}+n_{g,2}^{r}+n_{b,2}^{r}),\\
		n_c&=n_{c,1}+n_{c,2}+2(n_{g,1}^{r}+n_{b,1}^{r}+n_{g,2}^{r}+n_{b,2}^{r}),
		\label{eq:comp-growth-hz}
	\end{align}
	where $n^{r}$ counts nonzero generator columns, to the union of $j-1$ leaves and one further leaf, the left operand has $(j-1)g$ nonzero continuous and $j-2$ nonzero binary columns, since the slack factors have zero columns in $G^{\mathrm c}$. Hence the $j$th union adds $g+2(jg+j-2)$ continuous generators and $k+2(jg+j-2)$ equalities. Summing for $j=2,\dots,L$ gives $g_L$ and $k_L$. The union has $L$ nonempty leaves, each with $g_L$ continuous generators and $k_L$ equalities, and the counts follow from \eqref{eq:comp-growth-hz} complement counterpart $(9n_{gc}+n+n_c+1,\ 2n_{gc},\ 7n_{gc}+n_c+1)$ per leaf and the $(L-1)n$ equalities arising from the $L-1$ generalized intersections of the $L$ leaf complements. \qed
\end{pf}

\begin{rem}
	The cubic growth in Corollary~\ref{cor:comp-growth} is the consequence of two effects that compound: the hybrid zonotope union accumulates slack factors, and each of the $L$ leaves obtained by fixing the binary generators keeps all of them. It is not removed by a different nesting order, since a balanced tree still gives $\Theta(Lg\log L)$ generators per leaf. If the leaves of the hybrid zonotope are first reduced by eliminating factors forced to zero by the equality constraints, the hybrid zonotope route returns to leaf-native sizes with $9g$ in place of $3g$ continuous generators per leaf, but it then requires a presolve that has to detect the forced factors. In a Hybrid Cone this information is explicit in the cone constraints, so Lemma~\ref{lem:leaf-recovery} recovers the leaves exactly without any optimization. This means that, without the presolve, the complement of a Hybrid Cone with $L=16$ rotated boxes with $g=2$ generators each in $\mathbb{R}^{2}$ requires $144$ continuous generators, whereas the same set written as a hybrid zonotope through pairwise unions requires $124\,656$. The pairwise union of \cite{bird:complement} is the relevant baseline because a hybrid zonotope is most often obtained from a recursive procedure, such as reachability propagation or the successive union of obstacles. If the hybrid zonotope is instead constructed at once from $L$ known sets, the $L$-way union of \cite{glunt:sharp} reduces the growth from cubic to quadratic and the count to $9\,792$ continuous generators. This is still far above the $144$ of the Hybrid Cone, but it is the less common situation. 
\end{rem}

 \section{Simulation Results}
 \label{sec:applications}
 
 \subsection{Localization from Signal Strength Measurements}
 \label{ssec:simulationLocalization}
 
 We first evaluate the proposed Hybrid Cone representation on a planar localization example adapted from the set-valued state-estimation benchmark in \cite{siefert:estimationHZ}. In this example, the state represents the two-dimensional position of a moving target. The target evolves according to the discrete-time integrator dynamics
 $$
 x_{k+1}=x_k+u_k,
 $$
 where
 $$
 x_k,u_k\in\mathbb{R}^2 .
 $$
 The scalar measurement corresponds to the received signal power generated by four fixed sources. The noiseless measurement model is
 $$
 y =
 \sum_{i=1}^{4}
 \frac{1}{d_i^2+1},
 \qquad
 d_i=\|x-s_i\|_2,
 $$
 where the source locations are
 $$
 s_1=(1,3),\quad
 s_2=(-2,2),\quad
 s_3=(3,0),\quad
 s_4=(-1,-4).
 $$
 Thus, each source contributes a signal whose strength decays with the squared distance from the target, and the measured output is the sum of these contributions.
 
 The state space is restricted to the square domain
 $$
 \mathcal{X}=[-5,5]\times[-5,5],
 $$
 and the set-valued estimator is initialized with the prior
 $$
 \hat{\mathcal{X}}_{0|-1}=\mathcal{X}.
 $$
 The true initial condition and the known input sequence are chosen as
 $$
 x_0=
 \begin{bmatrix}
 	1\\
 	0
 \end{bmatrix},
 $$
 and
 $$
 \begin{bmatrix}
 	u_0 & u_1 & u_2 & u_3
 \end{bmatrix}
 =
 \begin{bmatrix}
 	-1 & -2 & -1 & 2\\
 	1 & -1 & -1 & -1
 \end{bmatrix}.
 $$
 
To encode the nonlinear measurement relation within the proposed Hybrid Cone framework, we replace the nonlinear signal-power map with a piecewise-affine surrogate obtained from a feedforward neural network. The network used in the simulations has two hidden layers, each containing twenty ReLU neurons. Since ReLU networks define piecewise-affine functions, the trained network admits an exact mixed discrete-continuous representation. The set-valued estimator will first propagate the previous estimate following
\begin{equation}
	\hat{\mathcal{X}}_{k+1 \mid k}=\hat{\mathcal{X}}_{k \mid k} \oplus\left\{u_k\right\},
\end{equation}
followed by an intersection with the set $\mathcal{Y}_{k+1}$ corresponding to the states conforming with the measurement $y_{k+1}$, i.e., 
\begin{equation}
	\hat{\mathcal{X}}_{k+1 \mid k+1}=\hat{\mathcal{X}}_{k+1 \mid k} \cap \hat{\mathcal{Y}}_{k+1},
\end{equation}
where
\begin{equation}
	\hat{\mathcal{Y}}_{k+1}=\begin{bmatrix}
		\mathbf{I}_2 & \mathbf{0}
	\end{bmatrix}\left(\Phi \cap_{[\mathbf{0}\ 1]}\left\{y_{k+1}\right\}\right),
	\
	\Phi=\Gamma_\phi\oplus\bigl(\{\mathbf 0_2\}\times[-\varepsilon,\varepsilon]\bigr),
\end{equation}
and $\Phi$ is the graph set for the ReLU NN assuming a value $a=5$ and $\epsilon$ is the maximal residual for $2 \times 10^{5}$ random validation samples. 

\begin{figure}
	\centering
	\includegraphics[clip=true, scale=.5,trim=50 300 10 300]{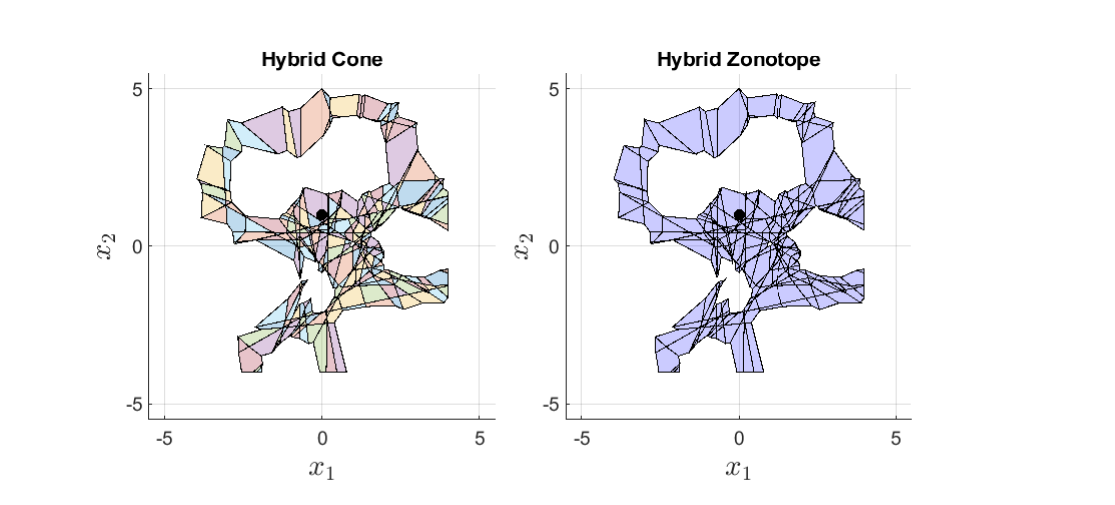}
	\caption{
		Set-valued localization estimate $\hat{\mathcal{X}}_{1|1}$ for both cases when using Hybrid Cones and Hybrid Zonotopes.
	}
	\label{fig:localization_k1}
\end{figure}

Figure~\ref{fig:localization_k1} compares the measurement-consistent localization set obtained at the first sampling time using Hybrid Cones and Hybrid Zonotopes. The two representations produce the same nonconvex feasible region, including the same disconnected components, holes, and boundary structure. This agreement is expected in this experiment because all variable bounds are polytopic, so the Hybrid Cone formulation does not yet exploit any genuinely conic constraints beyond what can already be represented by Hybrid Zonotopes. Moreover, each ReLU activation is naturally described, over a bounded domain, as the union of two affine line segments corresponding to its inactive and active regimes. Since Hybrid Zonotopes can exactly encode finite unions of polytopes through binary generators, the ReLU network measurement model admits an exact Hybrid Zonotope representation in this setting. Consequently, the Hybrid Cone and Hybrid Zonotope constructions are equivalent for this example.

\begin{figure}
	\centering
	\includegraphics[clip=true, scale=.5,trim=50 300 10 300]{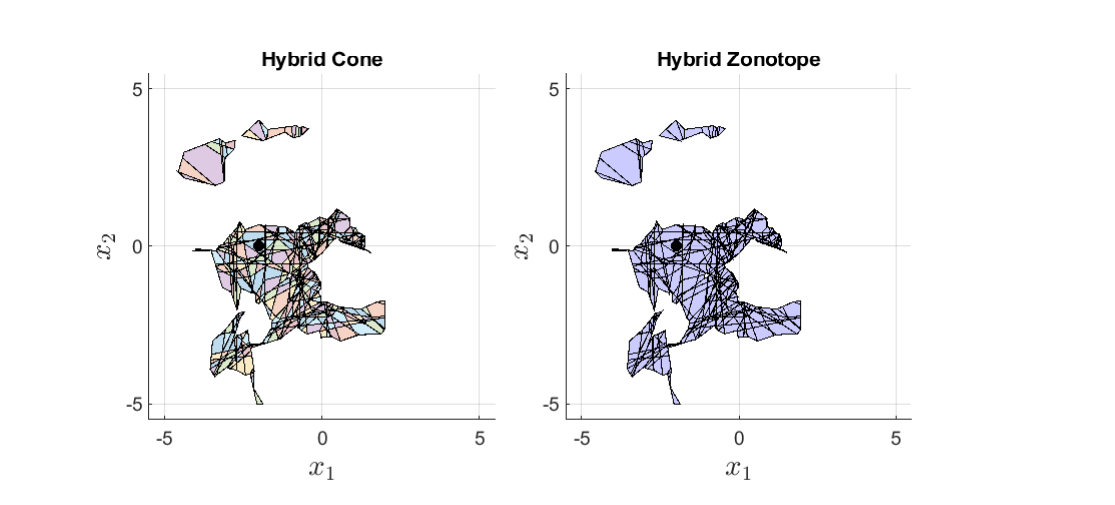}
	\caption{
		Set-valued localization estimate $\hat{\mathcal{X}}_{2|2}$ for both cases when using Hybrid Cones and Hybrid Zonotopes.
	}
	\label{fig:localization_k2}
\end{figure}

In Figure \ref{fig:localization_k2}, we can still observe that the produced sets remain the same for both representations. Although the plots show that Hybrid Zonotopes and Hybrid Cones produce the same reachable sets in the localization example, their representation sizes differ substantially. Table~\ref{tab:localization-representation} quantifies this difference over the first five iterations. The advantage of the proposed set representations stems from the more efficient description of the ReLU network graph. For the network with 40 ReLU neurons, the Hybrid Cone representation requires only 83 continuous generators (two per neuron, two for the network input, and one for the error bound), 40 binary generators (one per neuron), and 40 equality constraints to (one per neuron) to enforce the inter-layer connections. In contrast, the HZ representation requires 163 continuous generators, the same 40 binary generators, and 120 equality constraints. Specifically, each neuron requires four continuous generators and two additional constraints to encode its ReLU graph, in addition to the 40 equality constraints that enforce the inter-layer connections. Thus, HCs preserve the same nonconvex reachable sets while reducing the number of continuous generators by approximately 50\% and the number of equality constraints by approximately 66\% for the network graph. As shown in Table~\ref{tab:localization-representation}, these savings compound across subsequent iterations as we will apply the set describing the graph set one time per iteration.

\begin{table}
	\centering
	\renewcommand{\arraystretch}{1.15}
	\begin{tabular}{c|ccc|ccc}
		\hline
		& \multicolumn{3}{c|}{HZ (exact)}
		& \multicolumn{3}{c}{HC (exact)} \\
		$k$
		& $n_{\mathrm{gc}}$ & $n_{\mathrm{b}}$ & $n_{\mathrm{c}}$
		& $n_{\mathrm{gc}}$ & $n_{\mathrm{b}}$ & $n_{\mathrm{c}}$ \\
		\hline
		$0$ & $2$   & $0$   & $0$
		& $2$   & $0$   & $0$ \\
		$1$ & $165$ & $40$  & $123$
		& $85$  & $40$  & $43$ \\
		$2$ & $328$ & $80$  & $246$
		& $168$ & $80$  & $86$ \\
		$3$ & $491$ & $120$ & $369$
		& $251$ & $120$ & $129$ \\
		$4$ & $654$ & $160$ & $492$
		& $334$ & $160$ & $172$ \\
		\hline
	\end{tabular}
	\caption{Representation size of Hybrid Zonotopes and Hybrid Cones during set-based localization with a 40-neuron ReLU network. The entries report the numbers of continuous generators ($n_{\mathrm{gc}}$), binary generators ($n_{\mathrm{b}}$), and constraints ($n_{\mathrm{c}}$).}
	\label{tab:localization-representation}
\end{table}

\subsection{Saturated LQR: Reachability Analysis with Energy Constraints}
\label{subsec:saturated-lqr}

We next consider a saturated linear-quadratic regulator (LQR) example
based on the nonlinear reachability benchmark in
\texttt{zonoLAB} provided in \textrm{github.com/ESCL-at-UTD/zonoLAB}. The objective of this verification and
validation (V\&V) study is to assess the representation of admissible
initial conditions defined by energy constraints and to compare the
resulting reachable-set representations under the same closed-loop
dynamics. In this example, the initial set is nonconvex and consists of
the union of two ellipsoidal sets and one polytope. Hybrid Cones
can represent the ellipsoidal components directly through $\ell_2$-norm
constraints, thereby preserving the specified energy boundaries. In
contrast, the Hybrid Zonotope description approximates each
ellipsoidal component by an 8-sided polygon that contains the
corresponding energy-constrained set. This polygonal approximation
introduces conservatism into the initial-set representation, which can
propagate through subsequent reachability computations. The example
therefore highlights that beyond a savings in the number of generators and constraints, HCs
can encode a richer class of nonconvex sets in particular by retaining the geometry
of the original energy constraints.

\begin{figure}
	\centering
	\includegraphics[width=\linewidth]{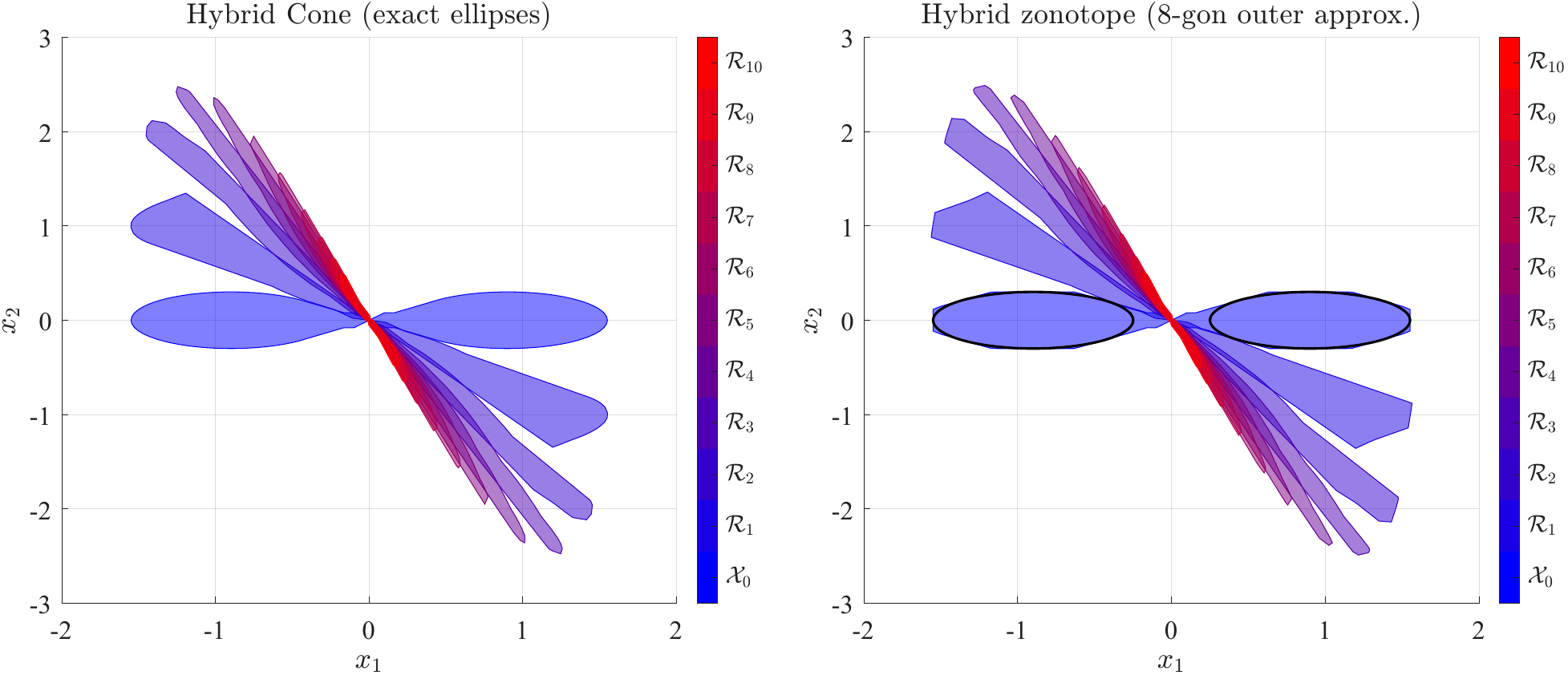}
	\caption{Reachable sets for the saturated LQR example using
		(left) Hybrid Cones, with exact ellipsoidal energy constraints, and
		(right) Hybrid Zonotopes, with 8-sided polygonal outer
		approximations of the ellipses.}
	\label{fig:saturated-lqr-reachability}
\end{figure}

Figure~\ref{fig:saturated-lqr-reachability} compares the reachable sets computed using the two representations. The left panel shows the HC results, in which the ellipsoidal energy constraints are represented exactly, whereas the right panel depicts the HZ results, in which the initial ellipses are enclosed by 8-sided polygonal outer approximations. The successive reachable sets are shown for the same sequence of time steps where is noticeable the conservatism being reduced by the fact that the closed-loop dynamics are stable. Nevertheless, if the objective was to certify some transient behavior of the closed loop, the analysis could produce an unsafe diagnostic caused by the added conservatism in a more general scenario or require a large increase in the number of generators to represent nonlinear boundaries as the union of additional polytopes. Thus, this example demonstrates the value of the HC representation even when the main advantage is not a substantial reduction in representation size, but the ability to express the specified nonconvex initial set without this additional geometric approximation.

Figure~\ref{fig:saturated-lqr-complexity} highlights the computational advantages of the proposed framework by plotting the growth of total generators ($n_g$) and equality constraints ($n_c$) over a $10$-step reachability horizon. While both set representations exhibit linear growth per time step, the Hybrid Cone formulation maintains a significantly leaner data structure throughout propagation. By step $10$, the HC representation requires $39.4\%$ fewer total generators ($n_g = 921$ compared to $1519$ for HZ) and achieves an $87.7\%$ reduction in equality constraints ($n_c = 80$ compared to $648$ for HZ). This constraint mitigation is a direct consequence of HC's native linear cone support and slack-free sharp union operations, which avoid converting linear inequalities into equality constraints via additional continuous slack variables. Combining the insights from Figures~\ref{fig:saturated-lqr-reachability} and~\ref{fig:saturated-lqr-complexity} demonstrates that Hybrid Cones do not trade off efficiency for expressiveness. They simultaneously eliminate geometric conservatism by exact representation of non-polytopic domains and maintain a substantially reduced constraint footprint over multi-step horizons, which reduces the frequency of order reduction methods. 

\begin{figure}
	\centering
	\includegraphics[width=.9\linewidth,trim=10 240 10 250]{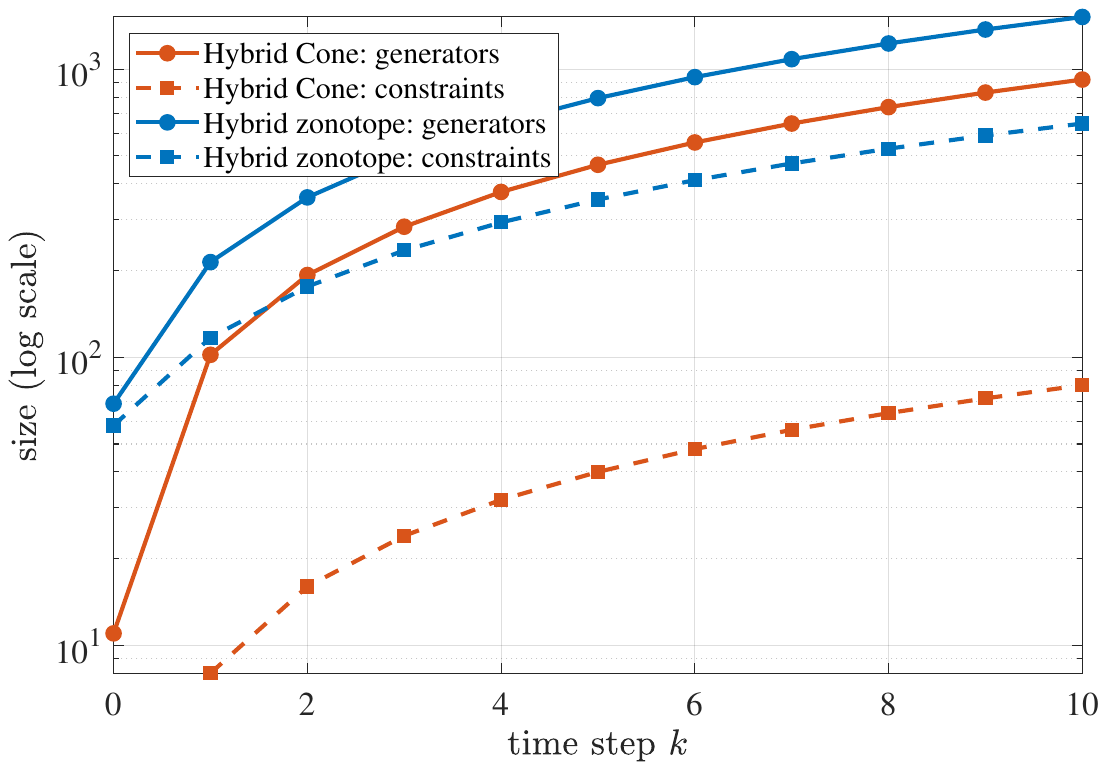}
	\caption{Representation complexity over the reachability time
		steps for Hybrid Cones and Hybrid Zonotopes. The numbers of
		generators and constraints (rows of $A$) are plotted on a
		logarithmic scale.}
	\label{fig:saturated-lqr-complexity}
\end{figure}

\section{Conclusion} \label{sec:conclusion}
This paper introduced Hybrid Constrained Convex Generators (hybrid CCGs) and their practical specialization, Hybrid Cones, as a unified framework for set-based analysis of nonlinear and hybrid systems. By embedding binary generator variables directly into the CCG architecture, Hybrid Cones capture complex nonconvex geometry including unions of polytopic and norm-cone sets while preserving the algebraic advantages of generator-based representations. 

The primary theoretical contribution is a sharp union construction derived from the exact closed-form convex hull of CCGs, which incorporates only one additional binary generator per fold without introducing artificial continuous slack variables or extra equality constraints. This structural efficiency extends directly to neural network reachability, enabling an exact representation of the single-neuron ReLU graph with two continuous generators, one binary generator, and zero equality constraints. Furthermore, the proposed leafwise complement operation eliminates the continuous slack variables required by hybrid zonotopes, reducing the representation growth from cubic $\mathcal{O}(L^3)$ to linear $\mathcal{O}(L)$ in the number of leaves when the original set is built from pairwise unions. 

Numerical benchmarks in set-valued localization and saturated LQR reachability showcase the reduction in the size of the data structure. The results demonstrated that Hybrid Cones eliminate geometric conservatism by directly incorporating non-polytopic energy constraints while achieving up to an $87.7\%$ reduction in equality constraints over multi-step horizons. In practice, this impacts the frequency at which order reductions have to be carried out to maintain the computational complexity.

Several promising directions for future research emerge from this work: i) extend order reduction algorithms from \cite{rego:comparison} specifically tailored for Hybrid Cones to bound generator and constraint accumulation over long propagation horizons, extending continuous zonotope reduction to hybrid cone structures; and, ii) applying Hybrid Cones to offline safety verification and online Model Predictive Control (MPC) for neural-network-controlled feedback systems, leveraging the compact graph encoding for faster mixed-integer solver execution.

\bibliographystyle{plainnat}
\bibliography{hybridLiterature}

\end{document}